\documentclass[journal]{IEEEtran}
\usepackage{cite}
\usepackage{graphicx}
\usepackage[hyphens]{url}

\usepackage{amsmath, amssymb, amsthm}
\usepackage{booktabs}
\usepackage{array}
\usepackage{threeparttable}
\usepackage{makecell}
\usepackage[version=4]{mhchem}
\theoremstyle{plain}\newtheorem{proposition}{Proposition}
\theoremstyle{remark}\newtheorem{remark}{Remark}
\newcommand{\Var}{\operatorname{Var}}
\newcommand{\BEP}{\mathrm{BEP}}
\newcommand{\dd}{\mathrm{d}}
\newcommand{\tb}{\bar{\theta}_\mathrm{B}}

\title{Autonomous Receptor Tuning via Integral Feedback for Adaptive Molecular Communication Receivers}
\author{Ya\u{g}{\i}z~Celhan and Murat~Kuscu,~\IEEEmembership{Senior Member,~IEEE}\thanks{The authors are with the Nano/Bio/Physical Information and Communications Laboratory (CALICO Lab), Department of Electrical and Electronics Engineering, Ko\c{c} University, Istanbul, T\"urkiye (e-mail: ycelhan24@ku.edu.tr; mkuscu@ku.edu.tr).}\thanks{This work was supported by the Scientific and Technological Research Council of Turkey (TUBITAK) under Grant \#123E516 and Grant \#123C592.}}

\begin{document}
\maketitle

\begin{abstract}
Molecular communication (MC) uses molecules as information carriers among engineered cells and artificial devices in the Internet of Bio-Nano Things. MC receivers commonly employ ligand receptors, whose discrimination between the two received concentration levels of binary signaling depends on their position relative to the dissociation constant. Changes in the transmitter-receiver distance, ligand degradation, and interference can shift both levels below or above the dynamic range of fixed receptors, impairing detection. Previous work restored performance by externally tuning the dissociation constant to the geometric mean of the estimated received levels. In this paper, inspired by receptor adaptation in living cells, we investigate a receiver whose intracellular \emph{antithetic integral feedback network} senses the receptor occupancy and autonomously tunes the dissociation constant through a modulator molecule. For constant received levels and equiprobable bits, we prove that, \emph{without knowing these levels}, the network sets the dissociation constant at equilibrium to their geometric mean, which minimizes the bit error probability (BEP) for any fixed receptor cooperativity. We derive the gain, bandwidth, noise, and stability boundary of this feedback loop. The random bits averaged by the network induce \emph{self-noise}, which sets an error floor independent of the network molecule count and, balanced against the channel-tracking error, yields a loop-speed design rule. For received levels varying fivefold over a coherence time of 24 loop time constants, adaptation reduces the BEP 61-fold for non-cooperative and 410-fold for cooperative receptors, compared with the best non-adaptive receiver. This benefit grows with the coherence time and peaks at moderate variation depths.
\end{abstract}


\section{Introduction}\label{sec:intro}
\IEEEPARstart{L}{iving} cells sense their environment through ligand receptors and continuously regulate the responses of these receptors. In bacterial chemotaxis, for instance, the receptor activity is reset as the background concentration of the attractant rises, which keeps the cell sensitive to fractional changes in concentration over five orders of magnitude~\cite{Mello2007effects,Kalinin2009logarithmic}. Similarly, eukaryotic cells shift, sharpen, or desensitize the response curves of their receptors, e.g., through methylation, phosphorylation, and allosteric modulation~\cite{Tu2018adaptation,Christopoulos2002allosterism}. At the heart of these mechanisms lies a common feedback structure, in which the cell measures the output of its own sensory system and moves the response curve toward the current signal level. In bacterial chemotaxis, this regulation implements integral feedback in the control-theoretic sense, which underlies the exactness of the adaptation~\cite{Yi2000robust,Alon1999robustness}.

Molecular communication (MC), in which information is encoded into the concentration of molecules, is envisioned as the physical layer of the Internet of Bio-Nano Things (IoBNT), comprising engineered cells and micro/nanomachines in biological environments~\cite{Akan2017fundamentals,Farsad2016survey,Akyildiz2008nanonetworks}. Similar to their natural counterparts, MC receivers based on engineered cells detect the transmitted symbols through ligand receptors, and the sensitive region of their response curve, i.e., their dynamic range, is centered at the dissociation constant $K_\mathrm{D}$ of these receptors~\cite{Kuscu2019transmitter}. In practice, the received concentration levels decrease with the transmitter-receiver distance and the enzymatic degradation of ligands, and increase with the background interference and the intersymbol interference (ISI) resulting from the long memory of diffusive channels. In our previous work~\cite{Kuscu2023adaptive}, we showed that a receiver with a fixed $K_\mathrm{D}$ suffers a severe loss in detection performance when the received levels leave its dynamic range, and that tuning $K_\mathrm{D}$ to the geometric mean of the two received levels restores the optimal performance. That analysis, however, as well as the adaptive detection methods proposed in the MC literature, treated adaptation as an external operation, assuming an estimate of the received levels and a means of tuning the receptor affinity accordingly.

Such external access does not exist for a receiver of cellular dimensions, as an engineered cell can act only through the chemical reactions it hosts. Therefore, adaptivity must be embedded in the receiver itself, through an intracellular chemical reaction network (CRN) that measures and corrects the detuning of the receptor affinity without channel state information (CSI), i.e., without any knowledge of the received levels. This is precisely the problem addressed in synthetic biology under the name of \emph{robust perfect adaptation}. A prominent solution to this problem is the \emph{antithetic integral feedback} (AIF) motif of Briat, Gupta, and Khammash~\cite{Briat2016antithetic}, in which two controller species annihilate each other in a one-to-one stoichiometry, and the surplus of one species over the other integrates the difference between a reference and a measurement, which returns the regulated quantity exactly to its setpoint under persistent disturbances. The AIF motif has been implemented in bacteria and mammalian cells~\cite{Aoki2019universal,Frei2022mammalian,Anastassov2023intein}, its performance limits and noise properties have been analyzed~\cite{Olsman2019hard,Briat2018pi,Filo2026sensor}, and cybergenetics provides a design methodology for such circuits~\cite{Khammash2022cybergenetics,Filo2023biomolecular,Khammash2021perfect}. In the control literature, however, the AIF motif regulates the abundance of a molecular species toward a biological target~\cite{Briat2016antithetic,Olsman2019hard,Filo2022hierarchy}, whereas in an adaptive MC receiver, the regulated quantity is a detection statistic, whose setpoint is determined by the detection problem.

In engineered MC systems, adaptivity has so far been addressed at the algorithmic level, e.g., through adaptive thresholds that follow the previously decoded bits~\cite{Damrath2016adaptive}, detection based on estimated channel impulse responses~\cite{Chang2018adaptive,Huang2020distance}, and a threshold controlled by a proportional-integral-derivative (PID) controller~\cite{ARTRx2025}, all of which are intended for electronic implementation. Reaction networks have also been proposed for maximum a posteriori detection~\cite{FAU2023mapcrn}, synchronization~\cite{FAU2024sync}, and interference cancellation~\cite{ChemSICal2025}. However, none of these studies applies feedback to the receptor parameters, and our previous work~\cite{Kuscu2023adaptive}, which identified integral-feedback CRNs as an enabling technology for tuning the receptor affinity, did not construct the feedback loop.

\begin{figure*}[t]
\centering
\includegraphics[width=\textwidth]{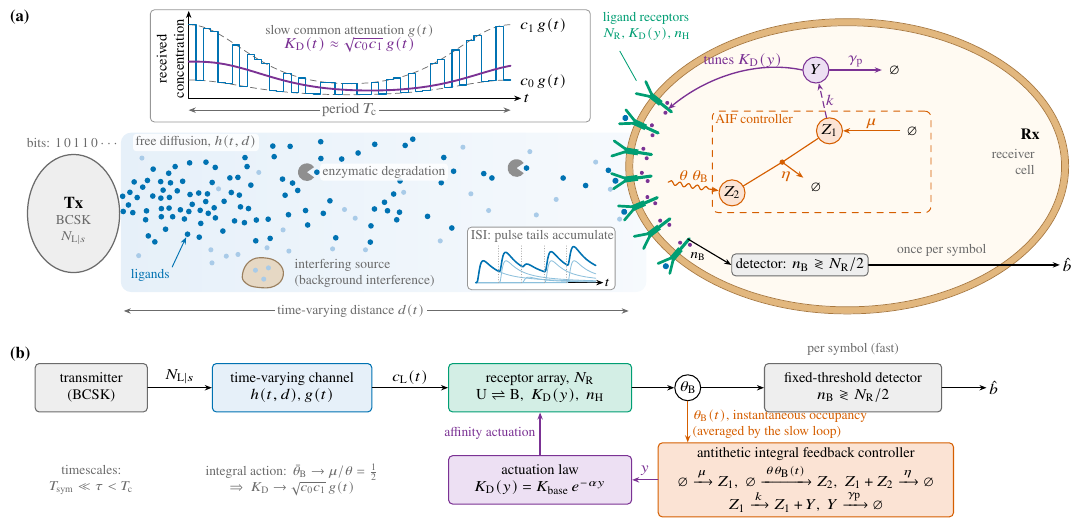}
\caption{(a)~Considered scenario. The transmitter (Tx) encodes the bits into the number of released ligands, which reach the receiver cell (Rx) by free diffusion. Changes in the transmitter-receiver distance and enzymatic degradation scale both received levels by a slow common attenuation $g(t)$ (upper inset, illustrative, with $c_0=1.0$, $c_1=4.0$, and a fivefold attenuation), an interfering source adds background molecules, and the tails of previous pulses accumulate as ISI (lower inset). Inside the Rx, the AIF controller compares the receptor occupancy, sensed through the production of $Z_2$, with the reference production of $Z_1$, and the modulator $Y$ tunes $K_\mathrm{D}(y)$ toward $\sqrt{c_0c_1}\,g(t)$. The detector counts the bound receptors once per symbol. (b)~Block diagram of the adaptive MC receiver, whose controller senses the instantaneous receptor occupancy and averages it over the symbols through its slow response.}
\label{fig:schematic}
\end{figure*}
In this paper, we construct this feedback loop and analyze its quantitative consequences. We investigate an MC receiver in which an AIF controller senses the receptor occupancy through the signaling pathway of the receptor, and tunes the $K_\mathrm{D}$ of the receptors through an intracellular modulator. The key observation is that, since the controller is slow compared to the symbol duration, it regulates the \emph{average} occupancy of the receptors over the incoming bit stream, and holds it at a setpoint of one half. The average of the occupancies for the two bits equals one half only when these occupancies are placed symmetrically about it, which, for a receptor with a Hill-type response, occurs only when $K_\mathrm{D}$ equals the geometric mean of the two received levels. Therefore, the setpoint of the biochemical controller coincides with the optimum of the detection problem for any fixed receptor cooperativity, i.e., the design rule imposed externally in~\cite{Kuscu2023adaptive} \emph{emerges as the equilibrium of a four-reaction network}. Nevertheless, the same slow response makes the controller sensitive to the random fraction of bit-1 transmissions within its averaging window. We show that this data-induced \emph{self-noise} is the dominant noise source of the receiver at physiological copy numbers of the network molecules, and that its balance against the error in tracking the channel variations determines how fast the feedback loop should be.

The main contributions of this paper are summarized as follows, where the first three items constitute the central results, and the remaining items support them:
\begin{itemize}
\item We propose an MC receiver in which an intracellular AIF controller tunes the dissociation constant of the ligand receptors, and prove that, for a piecewise-constant received signal, equal priors, a fixed Hill coefficient, and a fixed-threshold detector with random tie-breaking, its equilibrium places the dissociation constant, without any CSI, at the geometric mean of the received levels, which is the exact global minimizer of the bit error probability (BEP) of this detector and makes its decision rule the maximum-likelihood rule (Sections~\ref{sec:model} and~\ref{sec:lock}).
\item We derive the gain, bandwidth, noise, exact stability boundary, and leakage offset of the feedback loop, and verify them by stochastic simulation. We show that the self-noise imposes an error floor that does not decrease with the copy numbers of the network molecules, and identify the number of symbols within the averaging window as the key design parameter (Sections~\ref{sec:gain}--\ref{sec:leaky} and~\ref{sec:res-static}).
\item We derive a leading-order design rule for the loop speed by balancing the self-noise against the error in tracking the channel variations, and compare it with simulations over a range of loop speeds, coherence times, and variation depths of the channel (Sections~\ref{sec:bep} and~\ref{sec:res-headline}).
\item We compare the adaptive receiver with the best non-adaptive receiver, optimized with knowledge of the channel statistics, and determine the operating region, in terms of the variation depth and coherence time of the channel, in which adaptation is beneficial (Section~\ref{sec:res-headline}).
\item We show that receptor cooperativity widens the separation between the occupancies for the two bits by more than its noise cost, and thus improves the BEP at every tested copy number of the network molecules. Cooperativity co-regulated by the same modulator, on the other hand, constrains the actuator gain and, through its dependence on the attenuation, slows down the loop on the time-varying channel, which forfeits the benefit of cooperativity (Sections~\ref{sec:coop}, \ref{sec:res-headline}, and~\ref{sec:res-static}).
\item We identify chemical realizations of the sensor, controller, and actuator, and show that, under a CIR-based model of the received signal, the setpoint of the controller is invariant under attenuation, and the benefit of adaptation grows as the ISI weakens. An illustrative mapping to physical units calls for post-translational actuation and slowly diluted controller species (Sections~\ref{sec:res-tracking}, \ref{sec:res-headline}, and~\ref{sec:bio}).
\end{itemize}

The remainder of this paper is organized as follows. In Sections~\ref{sec:model} and~\ref{sec:analysis}, we present the system model and analyze the closed loop, respectively. The numerical results are presented in Section~\ref{sec:results}, followed by the discussion of the biological feasibility and the limitations in Section~\ref{sec:bio}, and the conclusion in Section~\ref{sec:conclusion}.

\section{System Model}\label{sec:model}
The considered scenario and the adaptive MC receiver are illustrated in Fig.~\ref{fig:schematic}(a), and the corresponding block diagram is shown in Fig.~\ref{fig:schematic}(b). Ligand molecules released by the transmitter bind to a receptor array on the surface of the receiving cell, and the dissociation constant $K_\mathrm{D}$ of these receptors is set by an intracellular modulator species $Y$. The receptor occupancy $\theta_\mathrm{B}$ is read once per symbol by a threshold detector, which produces the decoded bit $\hat b$, and is continuously sensed by an intracellular controller, which responds to the \emph{symbol-averaged occupancy} $\tb$ and adjusts the modulator concentration. Capital letters denote molecular species, and lower-case letters denote their concentrations, which the system size $\Omega$ converts into molecule numbers, e.g., $Y=\Omega y$. Rates and times are nondimensional, with rates of order unity, and one time unit is mapped to 25~s in Section~\ref{sec:bio}. Analytical expressions are written for $\Omega=1$.

\subsection{MC Channel and Received Signal}
We consider binary concentration shift keying (BCSK) over a diffusion-based MC channel, in which the transmitter releases $N_{\mathrm{L}|s}$ ligand molecules to transmit bit $s\in\{0,1\}$. The channel impulse response (CIR) of free diffusion is $h(t,d)=(4\pi Dt)^{-3/2}\exp(-d^2/4Dt)$, where $d$ is the transmitter-receiver distance and $D$ is the diffusion coefficient of the ligands. For sampling at the peak time $t_\mathrm{peak}=d^2/6D$ of the CIR, the received concentration for bit $s$ is $c_s=N_{\mathrm{L}|s}\,h(t_\mathrm{peak},d)$, up to an interference term discussed in Section~\ref{sec:res-tracking}~\cite{Kuscu2023adaptive}. The receiver must then distinguish between the two \emph{received levels} $0<c_0<c_1$, whose ratio $r=c_1/c_0$ is the \emph{level ratio}. On--off signaling, i.e., $c_0=0$, admits no finite affinity that places the average occupancy at one half, and is therefore outside the scope of the present design. The transmitted bits are independent and equiprobable, and each symbol lasts $T_\mathrm{sym}$, during which the received concentration is modeled as constant at $c_s$. We refer to this as the \emph{piecewise-constant signal model}. The shape of the diffusive pulse within a symbol, which is sensed by the controller but not by the detector, is treated with a CIR-based model in Section~\ref{sec:res-tracking}.

In practice, the received levels vary in time, since the transmitter-receiver distance may change, enzymes in the channel may degrade the ligands, and the output of the transmitter may fluctuate. Each of these effects scales both received levels by a common attenuation factor $g(t)$, i.e., $c_s(t)=c_s\,g(t)$, which is slow compared to the symbol duration, and is the time variation that the receiver must track. The change of the pulse shape with the transmitter-receiver distance is not considered here. The channel remains approximately constant over a coherence time $T_\mathrm{c}\gg T_\mathrm{sym}$~\cite{Kuscu2023adaptive,Ahmadzadeh2018mobile}, which we identify in Section~\ref{sec:results} with the interval between the steps of a step-wise attenuation or with the period of a periodic attenuation.

\subsection{Receiver and Detection}
The receiver senses the received concentration using $N_\mathrm{R}$ receptors, which reversibly bind the ligands through the reaction $\ce{U <=>[$c_\mathrm{L} k_+$][$k_-$] B}$, where U and B denote the unbound and bound states, $c_\mathrm{L}$ is the ligand concentration, and $k_+$ and $k_-$ are the binding and unbinding rates, respectively. The dissociation constant $K_\mathrm{D}=k_-/k_+$ gives a measure of the ligand-receptor affinity. Assuming that the binding reaction equilibrates within a symbol, the fraction of bound receptors follows the ligand concentration instantaneously through the Hill function
\begin{equation}
\theta_\mathrm{B}(c;K_\mathrm{D},n_\mathrm{H})=\frac{c^{n_\mathrm{H}}}{c^{n_\mathrm{H}}+K_\mathrm{D}^{n_\mathrm{H}}},
\label{eq:hill}
\end{equation}
where $n_\mathrm{H}$ is the Hill coefficient, which equals one for independent binding sites and exceeds one for cooperative receptors, in which case $N_\mathrm{R}$ counts independently switching units, such as receptor dimers or clusters. As the response curve flattens outside the dynamic range around $K_\mathrm{D}$, two received levels on the same side of this range produce nearly equal occupancies, and thus, the position of $K_\mathrm{D}$ relative to the received levels determines the error performance.

Detection is performed by a \emph{fixed-threshold detector}, whose threshold is set at half occupancy. Once per symbol, the detector counts the bound receptors, $n_\mathrm{B}\sim\mathrm{Binomial}(N_\mathrm{R},\theta_\mathrm{B}(c_s))$, and decides $\hat b=1$ if $n_\mathrm{B}>N_\mathrm{R}/2$ and $\hat b=0$ if $n_\mathrm{B}<N_\mathrm{R}/2$, while a tie, i.e., $n_\mathrm{B}=N_\mathrm{R}/2$, is broken at random, which, for even $N_\mathrm{R}$, is equivalent to the strict majority rule of $N_\mathrm{R}-1$ receptors. With equal priors, the BEP is given by
\begin{equation}
\begin{aligned}
\BEP&=\tfrac12\,P_{\mathrm{e}|0}+\tfrac12\,P_{\mathrm{e}|1},\qquad m=N_\mathrm{R}/2,\\
P_{\mathrm{e}|0}&=\Pr[n_\mathrm{B}>m\,|\,\theta_0]+\tfrac12\Pr[n_\mathrm{B}=m\,|\,\theta_0],\\
P_{\mathrm{e}|1}&=\Pr[n_\mathrm{B}<m\,|\,\theta_1]+\tfrac12\Pr[n_\mathrm{B}=m\,|\,\theta_1],
\end{aligned}
\label{eq:bep}
\end{equation}
where $\theta_s=\theta_\mathrm{B}(c_s)$ denotes the occupancy for bit $s$. For $n_\mathrm{H}=1$, we showed in~\cite{Kuscu2023adaptive} that the error probability of a threshold detector with an optimized threshold is minimized at $K_\mathrm{D}=\sqrt{c_0c_1}$, where the two occupancies are placed symmetrically about one half, and in Section~\ref{sec:lock}, we establish the exact optimality of this affinity for the fixed-threshold detector in~(\ref{eq:bep}) and an arbitrary $n_\mathrm{H}$. The threshold is kept fixed by design, since an adaptive threshold would make it impossible to separate the contribution of the receptors from that of the detector.

\subsection{Feedback Controller}
To render the receptor affinity adaptive, we let $K_\mathrm{D}$ depend on the concentration $y$ of an intracellular modulator species $Y$, as described in Section~\ref{sec:actuator}, which is, in turn, regulated by the AIF motif~\cite{Briat2016antithetic}. A controller species $Z_1$, representing the reference, is produced at a constant rate $\mu$, and a second controller species $Z_2$, representing the measurement, is produced at a rate proportional to the receptor occupancy. The two species annihilate each other in a one-to-one stoichiometry. Hence, the surplus of $Z_1$ over $Z_2$ is the integral of the error between the reference and the measurement, and as $Z_1$ catalyzes the production of the modulator, $Y$ is driven upward or downward until this error vanishes. The corresponding reactions are
\begin{equation}
\begin{alignedat}{2}
\varnothing&\xrightarrow{\;\mu\;}Z_1,\qquad\qquad&
\varnothing&\xrightarrow{\;\theta\,\theta_\mathrm{B}\;}Z_2,\\
Z_1+Z_2&\xrightarrow{\;\eta\;}\varnothing,&
Z_1&\xrightarrow{\;k\;}Z_1+Y,\\
Y&\xrightarrow{\;\gamma_\mathrm{p}\;}\varnothing,&
Z_1,\,Z_2&\xrightarrow{\;\gamma_\mathrm{c}\;}\varnothing,
\end{alignedat}
\label{eq:aif}
\end{equation}
where $\theta$ is the sensing rate, $\eta$ the annihilation (sequestration) rate, $k$ the actuation rate, $\gamma_\mathrm{p}$ the turnover rate of the modulator, and $\gamma_\mathrm{c}$ the dilution or degradation rate of the controller species. The controller comprises the first four reactions in~(\ref{eq:aif}), while the fifth reaction, the turnover of the modulator, completes the ideal closed-loop model, in which the controller-leak reactions are absent. The integrator is ideal for $\gamma_\mathrm{c}=0$ and leaky for $\gamma_\mathrm{c}>0$, which is analyzed in Section~\ref{sec:leaky}. The mean-field dynamics of the closed loop are then given by
\begin{equation}
\begin{aligned}
\dot z_1&=\mu-\eta z_1z_2-\gamma_\mathrm{c}z_1,\\
\dot z_2&=\theta\,\theta_\mathrm{B}\!\big(c_\mathrm{L}(t);K_\mathrm{D}(y),n_\mathrm{H}\big)-\eta z_1z_2-\gamma_\mathrm{c}z_2,\\
\dot y&=kz_1-\gamma_\mathrm{p}y,
\end{aligned}
\label{eq:ode}
\end{equation}
where $c_\mathrm{L}(t)$ denotes the received ligand concentration. Since the loop is slow compared to the symbol duration, as required in Section~\ref{sec:timescales}, its slow response is governed by the average of the sensed occupancy over the symbols, which, for equal priors, is given by
\begin{equation}
\tb(y)=\tfrac12\,\theta_\mathrm{B}(c_0;K_\mathrm{D}(y),n_\mathrm{H})+\tfrac12\,\theta_\mathrm{B}(c_1;K_\mathrm{D}(y),n_\mathrm{H}).
\label{eq:g}
\end{equation}
The controller thus regulates the mean of the occupancies for the two bits, whereas the detector observes each of them separately. The symbol-by-symbol fluctuation of the occupancy around this average, which is also sensed by the controller, constitutes a data-induced noise source, referred to as \emph{self-noise}, which is analyzed in Section~\ref{sec:gain}.

\subsection{Modulation of the Dissociation Constant}\label{sec:actuator}
The affinity of a receptor can be changed by an intracellular modulator that shifts the free energy of ligand binding, as in the two-state model of allostery~\cite{Monod1965allosteric,Christopoulos2002allosterism,Canals2012mwc,Keymer2006chemosensing}. We consider two such actuation laws. In the first one, each covalent modification of the receptor, such as the methylation of a glutamate residue in a chemoreceptor, shifts the binding free energy by a fixed amount $a$ in units of the thermal energy, such that the dissociation constant depends exponentially on the modification level $m$ of the receptor, $K_\mathrm{D}=K_0e^{-am}$.

Here, $K_\mathrm{D}$ is a phenomenological ligand-affinity parameter. In bacterial chemoreceptors, methylation shifts the free energy of the kinase-activity response, whose half-point moves exponentially with the number of methylated residues~\cite{Tu2008timevarying,Mello2007effects}, while the ligand-binding curve itself shifts much less~\cite{Hathcock2023nonequilibrium}. Hence, this regulation motivates the free-energy description but does not establish the ligand-affinity actuator assumed here, whose finite-range and kinetic approximations are tested in Section~\ref{sec:res-headline}.

The modulator $Y$ in~(\ref{eq:aif}) acts as the modifier, an enzyme or effector whose activity sets the modification level through $\dot m=k_\mathrm{m}y(M-m)-k_\mathrm{d}m$, with $M$ modification sites, modification rate $k_\mathrm{m}$, and removal rate $k_\mathrm{d}$. When the modification is far from saturation, $m\ll M$, and its removal is fast compared to the loop, the modification level follows the modifier as $m\approx(k_\mathrm{m}M/k_\mathrm{d})\,y$, and the law becomes
\begin{equation}
K_\mathrm{D}(y)=K_\mathrm{base}\,e^{-\alpha y},\qquad \alpha=\frac{ak_\mathrm{m}M}{k_\mathrm{d}},
\label{eq:exp}
\end{equation}
where $K_\mathrm{base}=K_\mathrm{D}(0)$, and $\alpha$ is the \emph{actuator gain}, i.e., the sensitivity of $\ln K_\mathrm{D}$ to the modulator concentration. Since the exact quasi-steady relation is $m_\infty=Mk_\mathrm{m}y/(k_\mathrm{d}+k_\mathrm{m}y)$, i.e., $K_\mathrm{D}=K_\mathrm{base}\exp[-\alpha y/(1+\alpha y/aM)]$, this \emph{exponential actuation law} requires both fast modification kinetics and a modification range $aM$ that is large compared to the required shift $\ln(K_\mathrm{base}/K_\mathrm{D})$, which, for the calibration in Table~\ref{tab:params}, equals $1+\ln S$ at the deepest point of a channel variation with depth $S$, where the attenuation factor reaches $g=1/S$. Both conditions are tested in Section~\ref{sec:res-headline}, and unless otherwise stated, the exponential law is used in the analysis and the simulations. By convention, $y$ denotes the activity that raises the affinity, which, for the Tar receptor, is the demethylating activity, as methylation lowers its affinity.

In the second law, an effector $Y$ binds an allosteric site of the receptor with a dissociation constant $K_\mathrm{Y}$ for the ligand-free receptor and $K_\mathrm{Y}/\chi$ for the ligand-bound receptor, where $\chi>1$ is the cooperativity factor of the allosteric ternary-complex model~\cite{Christopoulos2002allosterism}. Summing the binding polynomial over the four states of the receptor yields
\begin{equation}
\begin{aligned}
K_\mathrm{D}(y)&=K_\mathrm{base}\,\frac{1+y/K_\mathrm{Y}}{1+\chi\,y/K_\mathrm{Y}},\\
\beta(y)&=-\frac{\dd\ln K_\mathrm{D}}{\dd y}=\frac{(\chi-1)K_\mathrm{Y}}{(K_\mathrm{Y}+y)(K_\mathrm{Y}+\chi y)},
\end{aligned}
\label{eq:mwc}
\end{equation}
which is bounded between $K_\mathrm{base}/\chi$ and $K_\mathrm{base}$, and is therefore referred to as the \emph{bounded actuation law}, with the local actuator gain $\beta(y)$. The local analysis in Section~\ref{sec:analysis} applies to both laws, with $\alpha$ replaced by $\beta(y^*)$.

The two laws differ in their range and in the constancy of their actuator gain. The target affinity must lie within the range of the actuator, which, for~(\ref{eq:mwc}), requires $K_\mathrm{base}/\chi<\sqrt{c_0c_1}\,g<K_\mathrm{base}$ for every attenuation $g$ of the channel. Moreover, $\beta(y)$ decreases monotonically from $(\chi-1)/K_\mathrm{Y}$ at $y=0$, and thus, the loop slows down as the modulator concentration moves away from its calibration value, while the actuator gain of~(\ref{eq:exp}) is constant. The consequences of these differences are quantified in Section~\ref{sec:res-headline}.

We also consider two variants of the receptor array. In a \emph{mixed array}, $N_\mathrm{t}=wN_\mathrm{R}$ receptors are tunable, and the remaining $N_\mathrm{f}=(1-w)N_\mathrm{R}$ receptors are fixed at $K_\mathrm{fix}$, as in the receptor-expression adaptive receiver (REAR) of~\cite{Kuscu2023adaptive}. The controller then senses the bound fraction of the whole array, the count $n_\mathrm{B}$ is the sum of two binomial variables, which is treated exactly, and near $g=1$, where $K_\mathrm{fix}=K_\mathrm{D}(y^*)$, the loop gain is $w$ times that of a fully tunable array. In the second variant, the cooperativity of the receptors is \emph{co-regulated} by the modulator, $n_\mathrm{H}(y)=n_\mathrm{base}+\kappa y$ with $n_\mathrm{base}=1.5$ and $\kappa=1.0$, which models the methylation-dependent coupling of chemoreceptors~\cite{Skoge2011dynamics,Sourjik2004functional}. The model parameters, including the received levels $c_0=1.0$ and $c_1=4.0$, are listed in Table~\ref{tab:params}, and their physical counterparts are discussed in Section~\ref{sec:bio}.
\begin{table}[t]
\caption{Model Parameters (Nondimensional Calibration) and Their Physical Counterparts}
\label{tab:params}
\centering\scriptsize
\setlength{\tabcolsep}{3pt}
\begin{tabular}{@{}l>{\raggedright\arraybackslash}p{2.3cm}>{\raggedright\arraybackslash}p{0.8cm}>{\raggedright\arraybackslash}p{3.5cm}@{}}
\toprule
symbol & meaning & value & physical counterpart \\
\midrule
$\mu,\ \theta$ & reference, sensing & 0.5, 1.0 & post-translational, $10^{-1}$ to $1$~s$^{-1}$ \\
$k,\ \gamma_\mathrm{p}$ & actuation, $Y$ turnover & 0.5, 0.5 & post-translational, $\gamma_\mathrm{p}=0.02$~s$^{-1}$ in Section~\ref{sec:bio} \\
$\eta$ & sequestration & 1.0 & strong binding \\
$\gamma_\mathrm{c}$ & controller dilution & 0 & dilution, $10^{-4}$ to $10^{-3}$~s$^{-1}$ \\
$K_\mathrm{base}$ & $K_\mathrm{D}$ at $y=0$ & $2e$ & sets $y^*=1/\alpha$ \\
$c_0,\ c_1$ & received levels, bits 0/1 & 1.0, 4.0 & level ratio 50 in \cite{Kuscu2023adaptive} \\
$T_\mathrm{sym}$ & symbol duration & 0.5 & $3\,t_\mathrm{peak}$, 12.5~s \\
$N_\mathrm{R}$ & number of receptors & 200 & $10^3$ to $10^5$ per cell \\
$\Omega$ & system size & 200 & 12 to 800 in Fig.~\ref{fig:results}(b) \\
$n_\mathrm{base},\ \kappa$ & co-regulated $n_\mathrm{H}(y)$ & 1.5, 1.0 & methylation-dependent coupling \\
$w,\ K_\mathrm{fix}$ & tunable fraction, fixed $K_\mathrm{D}$ & 0.6, 2.0 & fixed receptor population \\
$T_\mathrm{c}$ & coherence time (design) & 80 & $10^{2}$ to $10^{4}$~s \\
time unit & model time & 1 & 25~s (Section~\ref{sec:bio}) \\
\bottomrule
\end{tabular}
\end{table}

\subsection{Timescale Requirements}\label{sec:timescales}
The receiver operates correctly only if the symbol duration $T_\mathrm{sym}$, the loop time constant $\tau$, and the coherence time $T_\mathrm{c}$ of the channel are ordered as
\begin{equation}
T_\mathrm{sym}\ll\tau<T_\mathrm{c}.
\label{eq:window}
\end{equation}
The first inequality ensures that the controller averages over the symbols. A loop faster than a symbol would retune the receptor affinity within each symbol, drive the occupancy to the setpoint for both bits, and thus remove the occupancy difference on which the detection relies. A loop much slower than a symbol, on the other hand, holds the average occupancy at the setpoint, while each symbol moves the occupancy above or below the threshold. The controller then averages approximately $\tau/T_\mathrm{sym}$ symbols, and the fluctuation in the fraction of bit-1 symbols within this \emph{averaging window} constitutes the self-noise. The second inequality ensures that the loop settles before the channel changes. In Section~\ref{sec:bep}, we turn these inequalities into a quantitative rule for $\tau$.

\section{Analysis of the Closed Loop}\label{sec:analysis}
\subsection{Steady State}\label{sec:lock}
At a steady state of~(\ref{eq:ode}) with $\gamma_\mathrm{c}=0$, subtracting the two controller balances yields $\theta\,\tb(y^*)=\mu$. Therefore, the sensed occupancy is maintained at $\mu/\theta$ irrespective of the received levels, which is the robust perfect adaptation property of the motif~\cite{Briat2016antithetic}. It is this property that eliminates the need for CSI at the receiver, since the setpoint is a ratio of two production rates fixed at design time. We choose $\mu/\theta=1/2$, and the following proposition establishes that this choice places the receptor at the optimal operating point.

\begin{proposition}\label{prop:lock}
Consider the Hill receptor~(\ref{eq:hill}) with a fixed $n_\mathrm{H}>0$, a piecewise-constant received signal with levels $0<c_0<c_1$ and equal priors, an array of $N_\mathrm{R}$ independent units, and the fixed-threshold detector with random tie-breaking of~(\ref{eq:bep}). The setpoint condition $\tb(y^*)=\tfrac12$ on the symbol-averaged occupancy~(\ref{eq:g}) holds if and only if $K_\mathrm{D}(y^*)=\sqrt{c_0c_1}$. At this affinity, the two occupancies are
\begin{equation}
\theta_s=\frac{c_s^{n_\mathrm{H}/2}}{c_0^{n_\mathrm{H}/2}+c_1^{n_\mathrm{H}/2}},\qquad \theta_0+\theta_1=1,
\label{eq:centering}
\end{equation}
the fixed-threshold decision rule is the maximum-likelihood rule, and $K_\mathrm{D}=\sqrt{c_0c_1}$ is the unique global minimizer of the error probability~(\ref{eq:bep}) over $K_\mathrm{D}$.
\end{proposition}
\begin{proof}
Let $K=K_\mathrm{D}(y^*)$ and $n=n_\mathrm{H}$. The condition $\theta_\mathrm{B}(c_0)+\theta_\mathrm{B}(c_1)=1$ is equivalent to $c_0^n/(c_0^n+K^n)=K^n/(c_1^n+K^n)$, i.e., to $c_0^nc_1^n=K^{2n}$, which gives $K=\sqrt{c_0c_1}$. Substituting $K^n=(c_0c_1)^{n/2}$ into~(\ref{eq:hill}) and dividing by $c_s^{n/2}$ yields~(\ref{eq:centering}). Since $\theta_1=1-\theta_0$, the two binomial likelihoods are mirror images of each other, and their ratio equals one at $n_\mathrm{B}=N_\mathrm{R}/2$, which makes the fixed-threshold decision rule the maximum-likelihood rule. For the optimality, let $m'=\lceil N_\mathrm{R}/2\rceil$ and $H(p)=\Pr[n_\mathrm{B}>N_\mathrm{R}/2]+\tfrac12\Pr[n_\mathrm{B}=N_\mathrm{R}/2]$ for $n_\mathrm{B}\sim\mathrm{Binomial}(N_\mathrm{R},p)$. For odd $N_\mathrm{R}$, this is the majority response of $N_\mathrm{R}$ units, and for even $N_\mathrm{R}$, the random tie-breaking makes it the majority response of $N_\mathrm{R}-1$ units, such that in both cases $H(p)=I_p(m',m')$, the regularized incomplete beta function, with $H'(p)=[p(1-p)]^{m'-1}/B(m',m')$. Write $p_s=\sigma(u_s)$ with the logistic function $\sigma$, $u_0=-h-v$, $u_1=h-v$, $h=\tfrac{n}{2}\ln(c_1/c_0)$, and $v=n\ln\big(K_\mathrm{D}/\sqrt{c_0c_1}\big)$, and let $F(u)=H(\sigma(u))$. Then, $F(-u)=1-F(u)$, and the error probability~(\ref{eq:bep}) is $P_\mathrm{e}(v)=\tfrac12[F(-h-v)+1-F(h-v)]=\tfrac12[F(-h-v)+F(-h+v)]$. Its derivative is $\dd P_\mathrm{e}/\dd v=\tfrac12[F'(h-v)-F'(h+v)]$, where $F'(u)=[\sigma(u)(1-\sigma(u))]^{m'}/B(m',m')$ is even and strictly decreasing in $|u|$. For $v>0$, $|h-v|<|h+v|$, and the derivative is positive, whereas for $v<0$, it is negative. Hence, $v=0$, i.e., $K_\mathrm{D}=\sqrt{c_0c_1}$, is the unique global minimum.
\end{proof}

Therefore, holding the average occupancy at one half places the two occupancies at equal distances above and below the threshold, i.e., at a \emph{symmetric operating point}, which minimizes the error probability of the fixed-threshold detector for any $n_\mathrm{H}$. With the actuation law~(\ref{eq:exp}), the corresponding modulator concentration is $y^*=\ln\big(K_\mathrm{base}/\sqrt{c_0c_1}\big)/\alpha$, which equals $1/\alpha$ for the calibration in Table~\ref{tab:params}. For $n_\mathrm{H}=1$, this result recovers the optimum reported in~\cite{Kuscu2023adaptive}, whereas for cooperative receptors, the optimum was not previously available in closed form. Two qualifications should be noted. First, the controller holds only the average occupancy, while within a symbol, the detector observes $\theta_0$ or $\theta_1$, i.e., the full occupancy separation, as verified in Section~\ref{sec:res-tracking}. Second, in the stochastic loop, provided that the closed loop is ergodic with finite moments~\cite{Briat2016antithetic}, the setpoint holds for the expectation of the sensed occupancy, and not for the affinity itself. For small Gaussian fluctuations, the mean affinity then differs from $\sqrt{c_0c_1}$ by the relative Jensen term $\alpha^2\Var(\delta y)/2$, which amounts to $0.9\,\%$ at $\Omega=200$, and is consistent with the measured mean affinity of $2.016$ reported in Section~\ref{sec:res-hc}.

\subsection{Receptor Gain, Bandwidth, and Modulator Noise}\label{sec:gain}
The dynamics of the loop around the steady state are governed by the \emph{receptor gain} $\zeta^*=\dd\tb/\dd y$ evaluated at $y^*$, i.e., the slope of the sensed occupancy with respect to the modulator concentration. From~(\ref{eq:hill}) and the actuation law~(\ref{eq:exp}), the slope of the occupancy at a single received level is $\partial\theta_\mathrm{B}(c)/\partial y=n_\mathrm{H}\alpha\,\theta_\mathrm{B}(1-\theta_\mathrm{B})$. At the steady state~(\ref{eq:centering}), the factor $\theta_\mathrm{B}(1-\theta_\mathrm{B})$ takes the same value $\theta_0\theta_1$ at both received levels, and therefore,
\begin{equation}
\zeta^*=n_\mathrm{H}\,\alpha\,\theta_0\theta_1,\qquad \theta_0=\frac{1}{1+r^{n_\mathrm{H}/2}}.
\label{eq:zeta}
\end{equation}

The controller dynamics comprise a fast and a slow part. When the sequestration mode is fast compared to the modulator turnover, $\eta z_\Sigma\gg\gamma_\mathrm{p}$ with $z_\Sigma=z_1^*+z_2^*$, which holds at the design points, e.g., with $\eta z_\Sigma=2.80$ for the non-cooperative receiver, the sum $z_1+z_2$ rapidly relaxes onto the manifold $z_1z_2=\mu/\eta$, and the slow variable becomes the difference $\varepsilon=z_1-z_2$, i.e., the integrated error. On this manifold, $z_1=\tfrac12\big(\varepsilon+\sqrt{\varepsilon^2+4\mu/\eta}\big)$, and the fraction $\rho\equiv\dd z_1/\dd\varepsilon=z_1^*/z_\Sigma$ quantifies the portion of a change in the error that appears in $Z_1$, and hence, in the actuation~\cite{Olsman2019hard}. Linearizing about the steady state yields the reduced loop
\begin{equation}
\delta\dot\varepsilon=-\theta\zeta^*\delta y-\gamma_\mathrm{c}\delta\varepsilon,\qquad \delta\dot y=k\rho\,\delta\varepsilon-\gamma_\mathrm{p}\delta y,
\label{eq:reduced}
\end{equation}
which, for $\gamma_\mathrm{c}=0$, is an integrator in series with a first-order process. The loop transfer function $L(s)=b/[s(s+\gamma_\mathrm{p})]$, with the loop gain $b=k\rho\theta\zeta^*$, has unit magnitude at the crossover frequency
\begin{equation}
\omega_\mathrm{c}^2=\tfrac12\Big(\sqrt{\gamma_\mathrm{p}^4+4b^2}-\gamma_\mathrm{p}^2\Big),
\label{eq:wc}
\end{equation}
which is the bandwidth of the receiver. For small changes in the logarithm of the attenuation, the reduced loop maps $\delta\ln g$ to $\delta\ln K_\mathrm{D}$ through the tracking transfer function $T(s)=b/(s^2+\gamma_\mathrm{p}s+b)$, and to the tracking error through $1-T(s)$, which satisfies $|1-T(j\omega)|\approx\omega\gamma_\mathrm{p}/b$ at low frequencies. The time constant of the slowest mode of~(\ref{eq:reduced}), $\tau=2/\big(\gamma_\mathrm{p}-\sqrt{\max(\gamma_\mathrm{p}^2-4b,0)}\big)$, which reduces to $\tau\approx\gamma_\mathrm{p}/b$ for a slow loop, is hereafter referred to as the \emph{loop time constant}. The low-frequency behavior of $1-T(s)$ defines a second timescale, the \emph{tracking time constant} $\tau_\mathrm{LF}=\gamma_\mathrm{p}/b$, which equals $\tau/[1-1/(\gamma_\mathrm{p}\tau)]$ in the overdamped regime, and is $12.6$ at the design points, compared to $\tau=10.1$. The two coincide only for a slow loop, and Section~\ref{sec:bep} uses $\tau_\mathrm{LF}$ wherever the distinction matters. Both the bandwidth and the loop time constant improve with increasing loop gain $b$. The slowest time constant of the full three-species loop is listed alongside $\tau$ in Table~\ref{tab:design}.

The loop gain also determines the noise level at which the receiver holds the steady state. The loop is driven by two noise sources, the first of which is the randomness of the reactions in~(\ref{eq:aif}). In the chemical Langevin description of the reduced loop~\cite{Gillespie2000langevin}, the noise driving the error has the intensity $\Gamma_\varepsilon=\mu+\theta\tb(y^*)=2\mu$, i.e., the total production rate of the two controller species, since the annihilation reaction removes one $Z_1$ and one $Z_2$ simultaneously, and thus, does not contribute to the noise of their difference~\cite{Briat2016antithetic}. The noise driving the modulator has the intensity $\Gamma_y=kz_1^*+\gamma_\mathrm{p}y^*=2\gamma_\mathrm{p}y^*$. The second noise source is the self-noise induced by the random transmitted bits. Within the averaging window of the loop, the sensing propensity switches by $\theta\Delta\theta$ with the bit, where $\Delta\theta=\theta_1-\theta_0$ is the occupancy separation. For independent and equiprobable bits, each lasting $T_\mathrm{sym}$, this switching has the covariance $(\theta\Delta\theta/2)^2(1-|t|/T_\mathrm{sym})$ for $|t|<T_\mathrm{sym}$, and for a loop that is slow compared to the symbol duration, its effect is equivalent to that of white noise with intensity
\begin{equation}
D_\mathrm{bits}=\frac{\theta^2\,\Delta\theta^2\,T_\mathrm{sym}}{4},
\label{eq:dbits}
\end{equation}
which enters the production of $Z_2$, and unlike the reaction noise, does not decrease with the system size. Solving the stationary Lyapunov equation of~(\ref{eq:reduced}) with the total error intensity $\Gamma_\varepsilon+\Omega D_\mathrm{bits}$ yields the variance of the modulator concentration as
\begin{equation}
\Var(\delta y)=\frac{1}{\Omega}\bigg[y^*+\frac{k\rho\,\mu}{\gamma_\mathrm{p}\theta\zeta^*}\bigg]+\frac{k\rho\,D_\mathrm{bits}}{2\gamma_\mathrm{p}\theta\zeta^*}.
\label{eq:vardy}
\end{equation}
The first term in~(\ref{eq:vardy}) is the birth-death noise of the modulator, the second term is the reaction noise of the controller, transmitted to the modulator and attenuated by the factor $\theta\zeta^*$, and the third term is the self-noise. The ratio of the third term to the second, $\Omega D_\mathrm{bits}/2\mu$, equals $2.8$ for $n_\mathrm{H}=1$ and $12.2$ for $n_\mathrm{H}=2.5$ at $\Omega=200$, which makes the self-noise the dominant noise source at physiological copy numbers of the network molecules. Its contribution decreases with the number of symbols within the averaging window, $\tau/T_\mathrm{sym}$, which is $20$ at the design points and $8$ for the fast co-regulated loop of Section~\ref{sec:coop}.

Raising the actuation rate $k$ at a fixed actuator gain increases the loop gain and, by~(\ref{eq:vardy}), the modulator variance, and thus trades faster tracking for larger adaptation noise. For the two co-regulated designs in Table~\ref{tab:design}, for instance, the full linear-noise model predicts $\Omega\Var(\delta y)=16.7$ at $k=0.500$ and $6.3$ at $k=0.129$, and the measured values are $14.3$ and $6.0$, respectively. Raising the actuator gain $\alpha$, on the other hand, lowers the modulator variance, but multiplies both the relative affinity variance, $\Var(\delta K_\mathrm{D})/K_\mathrm{D}^2=\alpha^2\Var(\delta y)$, and the noise at the detector input by $\alpha^2$. For this reason, the design in Section~\ref{sec:bep} compares the noise at the detector input. Finally, the full three-species Lyapunov solution, which retains the sequestration mode, exceeds~(\ref{eq:vardy}) by $3\,\%$ at the non-cooperative design point and by $16\,\%$ for the fast co-regulated loop. Table~\ref{tab:design} lists the full values, which agree with the simulations in Section~\ref{sec:res-hc} to within $1$ to $15\,\%$, depending on the design.
\begin{table}[t]
\centering
\begin{threeparttable}
\caption{Design Points and Performance Metrics of the Receivers}
\label{tab:design}
\scriptsize
\setlength{\tabcolsep}{2.4pt}\medmuskip=1mu
\begin{tabular}{@{}lccccc@{}}
\toprule
 & \makecell{non-\\cooperative} & mixed & cooperative & \makecell{co-regulated\\$\tau{=}4$} & \makecell{co-regulated\\$\tau{=}10$} \\
\midrule
$n_\mathrm{H}(y^*)$ & 1 & 1 & 2.5 & 2.50 & 2.50 \\
$\alpha$ & 0.384 & 0.776 & 0.257 & 1.000 & 1.000 \\
$k$ & 0.5 & 0.5 & 0.5 & 0.5 & 0.129 \\
$y^*$ & 2.60 & 1.29 & 3.89 & 1.00 & 1.00 \\
$\zeta^*$ & 0.085 & 0.103 & 0.082 & 0.319 & 0.319 \\
$\rho$ & 0.931 & 0.769 & 0.968 & 0.667 & 0.968 \\
$b=k\rho\theta\zeta^*$ & 0.0398 & 0.0398 & 0.0398 & 0.1064 & 0.0397 \\
$\omega_\mathrm{c}$ & 0.0785 & 0.0785 & 0.0785 & 0.1978 & 0.0785 \\
$\tau$ (full loop)\tnote{a} & 10.1\ (9.6) & 10.1\ (9.2) & 10.1\ (9.7) & 4.0\ (4.9) & 10.1\ (9.7) \\
$\Omega\Var(\delta y)$ &&&&& \\
\hspace{0.6em}chemical & 8.28 & 5.27 & 9.97 & 2.38 & 1.42 \\
\hspace{0.6em}self-noise & 15.5 & 10.7 & 73.4 & 14.3 & 4.9 \\
\hspace{0.6em}measured & 23.5 & 15.7 & 77.7 & 14.3 & 6.0 \\
$\zeta_\mathrm{crit}$ & 3.54 & 2.83 & 4.66 & 3.00 & 18.13 \\
$\zeta_\mathrm{crit}/\zeta^*$ & 41 & 27 & 57 & 9 & 57 \\
$\Delta\theta$ (design) & 0.333 & 0.333 & 0.700 & 0.700 & 0.700 \\
$\Delta\theta$ (measured) & 0.332 & 0.331 & 0.685 & 0.658 & 0.682 \\
$\sigma$ (measured)\tnote{b} & 0.0444 & 0.0442 & 0.0590 & 0.1037 & 0.0732 \\
molecules\tnote{c} & 1080 & 593 & 1580 & 500 & 1003 \\
BEP &&&&& \\
\hspace{0.6em}$\Omega{=}200$ & $1.7\times10^{-4}$ & $1.7\times10^{-4}$ & $2.5\times10^{-6}$ & $4.3\times10^{-3}$ & $3.1\times10^{-5}$ \\
\hspace{0.6em}averaged input\tnote{d} & $1.0\times10^{-5}$ & $8.9\times10^{-6}$ & $<10^{-8}$ & $2.9\times10^{-8}$ & $<10^{-8}$ \\
\hspace{0.6em}$\Omega{=}20$ & $1.1\times10^{-2}$ & $7.0\times10^{-3}$ & $1.4\times10^{-3}$ & $3.7\times10^{-2}$ & $1.1\times10^{-2}$ \\
\bottomrule
\end{tabular}
\begin{tablenotes}[flushleft]
\scriptsize
\item[] All entries are for $\Omega=200$ unless stated otherwise and use an ideal integrator, $\gamma_\mathrm{c}=0$. Variances are $\Omega$-scaled coefficients.
\item[a] Reduced-loop time constant, with the slowest time constant of the full three-species loop in parentheses.
\item[b] Measured bit-1 standard deviation for the co-regulated receivers, as their two symbol-conditioned output distributions have different variances.
\item[c] Controller and modulator molecules at $\Omega=200$.
\item[d] Controller sensing the symbol-averaged occupancy directly, i.e., without self-noise.
\end{tablenotes}
\end{threeparttable}
\end{table}

\subsection{Detection Variance, Bit Error Probability, and Loop Speed}\label{sec:bep}
The detector observes the occupancy, and not the modulator. For a receptor with a fixed $n_\mathrm{H}$, the occupancies at the two received levels have the same slope, $\partial\theta_\mathrm{B}(c_s)/\partial y=\zeta^*$ by~(\ref{eq:zeta}), and thus, both fluctuate with the variance $(\zeta^*)^2\Var(\delta y)$. Moreover, conditional on the modulator state and the transmitted symbol, the bound fraction of $N_\mathrm{R}$ independent receptor units has the variance $p_s(1-p_s)/N_\mathrm{R}$, with $p_s=\theta_\mathrm{B}(c_s;K_\mathrm{D}(y),n_\mathrm{H})$, which gives the sampling term $\theta_0\theta_1/N_\mathrm{R}$ at the symmetric operating point. This term characterizes the instantaneous detector studied here, and is not a universal limit, as a detector with temporal integration would require a model incorporating the receptor correlation times~\cite{Berg1977physics,Bialek2005physical,Kaizu2014berg}. Within the fixed-$n_\mathrm{H}$ local approximation used here, the variance at the detector input is the same for both symbols, and is given by
\begin{equation}
\sigma^2=(\zeta^*)^2\,\Var(\delta y)+\frac{\theta_0\theta_1}{N_\mathrm{R}}.
\label{eq:sigma}
\end{equation}
At high copy numbers of the network molecules, where the occupancy fluctuations are approximately Gaussian, the error probability can be approximated as $Q(\Delta\theta/2\sigma)$, where $Q(\cdot)$ is the Gaussian Q-function. We use this Gaussian expression only to locate the design points and to derive the loop-speed rule below. All error probabilities reported in this paper are computed from the exact expression~(\ref{eq:bep}), averaged over the simulated modulator fluctuations. Stochastic estimates below $10^{-8}$ depend on rare modulator excursions that the ensemble does not resolve, and are therefore reported only as lying below this \emph{reporting threshold}. Analytically evaluated benchmarks, such as the error probability of a non-adaptive receiver, on the other hand, are reported as computed.

Equations~(\ref{eq:vardy}) and~(\ref{eq:sigma}) have three consequences that determine the design. First, on a static channel, a larger receptor gain only increases the noise at the detector input, since the controller and self-noise terms of $\sigma^2$ increase with $\zeta^*$. Second, a non-adaptive receptor with $K_\mathrm{D}$ fixed at $\sqrt{c_0c_1}$ carries neither controller noise nor self-noise, and thus, on a static channel, it outperforms every adaptive receiver, by a factor of $270$ for $n_\mathrm{H}=1$ at $g=1$. For the independent and identically distributed (iid) symbol source and the feedback architecture studied here, adaptation carries a \emph{noise cost}, because the controller responds both to channel changes and to the fluctuations of its measurement, and whether the tracking benefit outweighs this cost is evaluated in Section~\ref{sec:res-headline}. Third, since the self-noise term does not decrease with $\Omega$, the error probability of the adaptive receiver flattens at high copy numbers, as the sweep in Section~\ref{sec:res-lc} indicates, toward a floor set by the length of the averaging window in symbols, and not by the copy numbers. This floor applies to the iid equiprobable symbol source, as source correlations, e.g., introduced by balanced coding, can change the low-frequency forcing of the loop, which is not evaluated here.

The loop speed is therefore a compromise between noise and tracking, analogous to the trade-off between misadjustment and tracking lag in adaptive filters~\cite{Haykin2014adaptive}, which can be turned into a design rule using~(\ref{eq:vardy}) together with $T(s)$. In the reduced white-noise model, the self-noise term of~(\ref{eq:sigma}) is exactly $\sigma_\mathrm{bits}^2=b\Delta\theta^2T_\mathrm{sym}/(8\gamma_\mathrm{p})=\Delta\theta^2T_\mathrm{sym}/(8\tau_\mathrm{LF})$, such that the signal-to-noise ratio (SNR) set by the self-noise, $\Delta\theta/\sigma_\mathrm{bits}=\sqrt{8\tau_\mathrm{LF}/T_\mathrm{sym}}$, depends only on the tracking time, and equals $14.2$ at the design points, independently of the cooperativity. With the calibration $y^*=1/\alpha$, and with $\rho$ frozen at its design value, the reaction terms of~(\ref{eq:sigma}) also scale as $1/\tau_\mathrm{LF}$, such that the loop-dependent part of the noise variance at the detector input becomes $A/\tau_\mathrm{LF}$, where $A=\Delta\theta^2T_\mathrm{sym}/8+[n_\mathrm{H}\theta_0\theta_1\gamma_\mathrm{p}/(k\rho\theta)+\mu/\theta^2]/\Omega$. The sampling term of~(\ref{eq:sigma}) does not depend on the loop speed. Note that $A$ is not a controller-independent constant, but a coefficient of the tested design family, since $\rho$ changes with $\alpha$. On the periodic channel $g(t)=\exp[-\tfrac12\ln S\,(1-\cos2\pi t/T_\mathrm{c})]$ of Section~\ref{sec:res-headline}, the affinity lags the time-varying optimum by $\delta\ln K_\mathrm{D}=\tfrac12\ln S\,|1-T(j\omega)|\cos\varphi$, where $\omega=2\pi/T_\mathrm{c}$ and $\varphi$ is the phase. This lag shifts the logit of both occupancies by $n_\mathrm{H}\delta\ln K_\mathrm{D}$, and the Gaussian approximation of the error probability becomes the phase average
\begin{equation}
\BEP\approx\Big\langle\tfrac12\,Q\Big(\frac{\tfrac12-\theta_0(\varphi)}{\sigma}\Big)+\tfrac12\,Q\Big(\frac{\theta_1(\varphi)-\tfrac12}{\sigma}\Big)\Big\rangle_\varphi,
\label{eq:surrogate}
\end{equation}
where $\theta_s(\varphi)$ denote the occupancies under the detuning at phase $\varphi$, and $\sigma$ is the static-channel value given by~(\ref{eq:sigma}). At low frequencies, the detuning amplitude is $\pi\ln S\,\tau_\mathrm{LF}/T_\mathrm{c}$, and its contribution to the variance of the occupancy is $B\tau_\mathrm{LF}^2$, where $B=(n_\mathrm{H}\theta_0\theta_1\pi\ln S)^2/(2T_\mathrm{c}^2)$. Minimizing $A/\tau_\mathrm{LF}+B\tau_\mathrm{LF}^2$ then yields the leading-order optimum
\begin{equation}
\tau_\mathrm{LF,opt}=\Big(\frac{A}{2B}\Big)^{1/3}\propto T_\mathrm{c}^{2/3}\,(\ln S)^{-2/3},
\label{eq:tauopt}
\end{equation}
which grows with the coherence time and shrinks with the variation depth. The design points in Table~\ref{tab:design} were located using the crossover condition $\omega_\mathrm{c}\ge2\pi/T_\mathrm{c}$ at the design coherence time $T_\mathrm{c}=80$, which places them at $\tau_\mathrm{LF}=12.6$. The candidate~(\ref{eq:tauopt}), however, is unconstrained. In a design family in which only the actuator gain $\alpha$ varies, the tracking time $\tau_\mathrm{LF}=\gamma_\mathrm{p}/b(\alpha)$ has a minimum, which, for the non-cooperative family, is $\tau_\mathrm{LF,min}=6.36$ at $\alpha=1.41$, and the implementable candidate is $\max\{\tau_\mathrm{LF,min},(A/2B)^{1/3}\}$. As~(\ref{eq:tauopt}) is a leading-order noise-tracking rule, and not an exact minimizer of the error probability, we compare the Gaussian approximation~(\ref{eq:surrogate}) and the projected rule with the full nonlinear receiver in Section~\ref{sec:res-headline}.

\subsection{Stability}\label{sec:stab}
As a larger receptor gain increases the bandwidth, we determine the largest receptor gain $\zeta$ for which the loop remains stable. Although the reduced loop~(\ref{eq:reduced}) is stable for any $\zeta$, the relevant question concerns the full loop~(\ref{eq:ode}), including the fast sequestration mode eliminated by the reduction. The Jacobian of the full loop at the steady state, in the variables $(z_1,z_2,y)$, has the characteristic polynomial $s^3+a_2s^2+a_1s+a_0$, with $q=\gamma_\mathrm{c}(\eta z_\Sigma+\gamma_\mathrm{c})$ and
\begin{equation}
\begin{aligned}
a_2&=\eta z_\Sigma+2\gamma_\mathrm{c}+\gamma_\mathrm{p},\qquad a_1=q+\gamma_\mathrm{p}(\eta z_\Sigma+2\gamma_\mathrm{c}),\\
a_0&=\gamma_\mathrm{p}q+k\theta\zeta\,\eta z_1^*.
\end{aligned}
\end{equation}
The Routh--Hurwitz condition $a_2a_1>a_0$ gives the exact local stability boundary of the deterministic loop as
\begin{equation}
\zeta<\zeta_\mathrm{crit}=\frac{a_2a_1-\gamma_\mathrm{p}q}{k\theta\eta z_1^*}
\;\xrightarrow{\;\gamma_\mathrm{c}=0\;}\;
\frac{\gamma_\mathrm{p}\,(\gamma_\mathrm{p}+\eta z_\Sigma)}{k\theta\rho}.
\label{eq:zetacrit}
\end{equation}
The boundary~(\ref{eq:zetacrit}) differs from the production-degradation criterion of~\cite{Olsman2019hard}, which is derived in the strong-binding limit for a process with an additional production-degradation stage between the actuator and the sensed species, whose phase lag allows a limit cycle to form. In the present receiver, the sensed occupancy is an instantaneous function of the modulator, and the additional dynamics are those of the sequestration mode, whose rate $\eta z_\Sigma$ enters~(\ref{eq:zetacrit}) linearly. Hence, the loop is unconditionally stable in the strong-binding limit. As a reference in Fig.~\ref{fig:results}(c), we apply the bound of~\cite{Olsman2019hard} on the product of the production rates around the loop to the loop gain $b=k\rho\theta\zeta$, i.e., $S_\mathrm{f}=[k\rho\theta\zeta/(2\gamma_\mathrm{p}(\gamma_\mathrm{p}+\gamma_\mathrm{c})^2)]^{1/3}$, and refer to $S_\mathrm{f}=1$ as the \emph{normalized two-stage reference}. This reference is not a dynamically equivalent mapping between the two architectures, and for the $\alpha=1$ illustration, it lies at $\zeta=0.250\,\zeta_\mathrm{crit}$, i.e., a factor of $4$ below the boundary~(\ref{eq:zetacrit}).

We test~(\ref{eq:zetacrit}) in Fig.~\ref{fig:results}(c) by simulating the loop at the steady state for $\alpha=1$ with the logistic sensing function $\tb=1/\{1+\exp[-4\zeta(y-y^*)]\}$, which has the slope $\zeta$ at $y^*$, and allows $\zeta$ to be varied independently of the receptor map. For two values of $\Omega$, we record the excess kurtosis of the modulator fluctuation and the modulator variance propagated to the occupancy through the slope, $\Omega\zeta^2\Var(\delta y)$, which the linear noise analysis predicts to be independent of $\Omega$. Up to about half of $\zeta_\mathrm{crit}$, the excess kurtosis is close to zero, and the scaled variances coincide, as expected for a Gaussian fluctuation whose variance scales with $1/\Omega$. As $\zeta_\mathrm{crit}$ is approached and exceeded, the scaled variance ceases to be independent of $\Omega$, reaching $354$ at $\zeta=2\zeta_\mathrm{crit}$, and the excess kurtosis falls toward $-1.5$, i.e., the value for a sinusoid, which is the signature of a limit cycle. Both signatures develop around $\zeta_\mathrm{crit}$, and the reference $S_\mathrm{f}=1$, marked by the dashed line, does not predict the instability boundary. Note that this boundary is a local property of the deterministic loop, and the oscillation beyond it is established only by simulation. At the design points in Table~\ref{tab:design}, the ratio $\zeta_\mathrm{crit}/\zeta^*$ is at least $9$, and its smallest value over the entire range of actuator gains considered is $8.4$. Therefore, the receivers are far from instability, and the relevant design trade-off is between speed and noise.

\subsection{Receptor Cooperativity}\label{sec:coop}
According to~(\ref{eq:centering}), the occupancy separation depends only on the level ratio and the cooperativity, $\Delta\theta=(r^{n_\mathrm{H}/2}-1)/(r^{n_\mathrm{H}/2}+1)$. For $r=4.0$, it increases from $0.333$ at $n_\mathrm{H}=1$ to $0.700$ at $n_\mathrm{H}=2.5$, while the sampling term $\theta_0\theta_1/N_\mathrm{R}$ decreases, since the two occupancies move toward zero and one, respectively. Its receptor gain~(\ref{eq:zeta}) also changes, but the design procedure re-selects $\alpha$ such that the bandwidth meets the tracking requirement $\omega_\mathrm{c}\ge2\pi/T_\mathrm{c}$ with equality, and the two receivers then share the same loop gain $b$, as listed in Table~\ref{tab:design}. This advantage, however, is not free of cost. At the detector, the controller noise of the cooperative receiver is $11\,\%$ higher than that of the non-cooperative receiver, and its self-noise is $4.4\times$ higher, since $D_\mathrm{bits}$ grows with $\Delta\theta^2$. Moreover, since $\sigma_\mathrm{bits}^2/\Delta\theta^2=T_\mathrm{sym}/(8\tau_\mathrm{LF})$ at matched loop speed, cooperativity does not improve the SNR set by the self-noise. Its benefit rather comes from the receptor sampling term, which decreases as the separation increases, and from the exact tails of~(\ref{eq:bep}). This benefit outweighs the noise cost, such that the cooperative receiver has the lowest error probability at every copy number considered in Section~\ref{sec:results}, by a factor of $66$ at $\Omega=200$.

The situation is different when the cooperativity is co-regulated by the same modulator, $n_\mathrm{H}(y)=n_\mathrm{base}+\kappa y$. In this case, the slopes of the two occupancies become $\partial\theta_\mathrm{B}(c_s)/\partial y=\theta_0\theta_1\,[n_\mathrm{H}\alpha\pm\tfrac12\kappa\ln r]$, with the negative sign for bit 0. The $\kappa$ term cancels in the average and leaves the local loop gain unchanged, but the two symbol-conditioned slopes differ, such that the fluctuations at the detector input must be assessed separately for the two symbols. Co-regulation also constrains the design. Since $n_\mathrm{H}(y^*)=n_\mathrm{base}+\kappa y^*$ with $y^*=1/\alpha$ ties the equilibrium cooperativity to the actuator gain, a comparison with the cooperative receptor at equal cooperativity, $n_\mathrm{H}(y^*)=2.5$, fixes $\alpha=1$, and hence, the receptor gain $\zeta^*=0.319$. With the remaining controller rates at their calibration values, this receptor gain results in a loop gain of $b=0.1064$, compared to $0.0398$ for the non-cooperative design, a loop time constant of $4.0$, and an averaging window of only $8$ symbols, such that the self-noise dominates the error probability, as shown in Section~\ref{sec:res-hc}. The constraint acts on $\zeta^*$, not on $b$. Reducing the actuation rate to $k=0.129$ restores the loop gain of the design, giving $\tau=10.1$, while $n_\mathrm{H}(y^*)=2.5$ remains unchanged, and this bandwidth-matched co-regulated design is listed next to the fast one in Table~\ref{tab:design}. Under an attenuation $g$ of the channel, the co-regulated equilibrium moves to $y^*=1-\ln g$, such that both the cooperativity $n_\mathrm{H}(y^*)=n_\mathrm{base}+\kappa(1-\ln g)$ and the receptor gain $\bar\zeta=n_\mathrm{H}(y^*)\alpha\theta_0\theta_1$ depend on the attenuation, whereas~(\ref{eq:zeta}) is invariant under attenuation for a fixed $n_\mathrm{H}$. This dependence becomes important on the time-varying channel of Section~\ref{sec:res-headline}.

\textit{Optimality under co-regulation.} When $n_\mathrm{H}$ varies with the modulator, the reference $\mu/\theta=1/2$ still imposes the symmetric operating point $\theta_0+\theta_1=1$ at an attainable deterministic equilibrium, but it does not generally minimize the error probability along the joint affinity--cooperativity actuator curve, and the optimal reference of the co-regulated receiver is left open.

\begin{remark}[Mixed arrays and cooperativity-only tuning]
In a mixed array with a tunable fraction $w>1/2$, the setpoint remains attainable for every occupancy $\theta_\mathrm{f}$ of the fixed population, since the sensed bound fraction spans $[(1-w)\theta_\mathrm{f},\,w+(1-w)\theta_\mathrm{f}]$, provided that the required affinity lies within the range of the actuator. Away from the calibration point, however, the tunable population compensates for the fixed one, and is pushed away from its own symmetric operating point. The occupancy separation of the array therefore shrinks. Tuning the cooperativity alone with a fixed $K_\mathrm{D}$, on the other hand, is not effective, since $\partial\theta_\mathrm{B}/\partial n_\mathrm{H}=\theta_\mathrm{B}(1-\theta_\mathrm{B})\ln(c/K_\mathrm{D})$ has opposite signs at the two received levels, and its average vanishes at the setpoint. Therefore, cooperativity is a design parameter of the receptor, and not a control variable of the loop.
\end{remark}

\subsection{Leaky Integration}\label{sec:leaky}
In a growing cell, the controller species are diluted, and the integrator becomes leaky~\cite{Qian2018leaky}. In this case, subtracting the two controller balances of~(\ref{eq:ode}) at steady state gives
\begin{equation}
\tb(y^*)=\frac{\mu}{\theta}+\frac{\gamma_\mathrm{c}\,(z_2^*-z_1^*)}{\theta},
\label{eq:leaky}
\end{equation}
which shows that the steady-state affinity moves away from the geometric mean of the received levels by an amount proportional to the leak rate and to the imbalance between the concentrations of the two controller species. For a small leak, the resulting detuning follows from the zero-leak state as $\delta\ln K_\mathrm{D}=\alpha\gamma_\mathrm{c}(z_1^*-z_2^*)/(\theta\zeta^*)$, which predicts the equilibrium affinity to within one percent up to $\gamma_\mathrm{c}=10^{-2}\gamma_\mathrm{p}$.

The effect of the leak is significant. For the non-cooperative receiver, a leak of $10^{-3}\gamma_\mathrm{p}$ increases the deterministic error probability by $1.02\times$, and a leak of $10^{-2}\gamma_\mathrm{p}$ by $2.85\times$, whereas a leak of $10^{-1}\gamma_\mathrm{p}$ moves the affinity by a factor of $1.40$ and increases the error probability by $4800\times$. A ten-percent increase in the noiseless error probability requires $\gamma_\mathrm{c}<2.3\times10^{-3}\,\gamma_\mathrm{p}$ for the non-cooperative receiver and $\gamma_\mathrm{c}<4.1\times10^{-4}\,\gamma_\mathrm{p}$ for the cooperative receiver, whose narrower error tails make it more sensitive to detuning. The relevant budget for the adaptive receiver, however, is its noisy error probability, which the detuning increases less steeply. The Gaussian approximation based on~(\ref{eq:sigma}), with the zero-leak variance held fixed, places the same penalty at $\gamma_\mathrm{c}=5.0\times10^{-3}\,\gamma_\mathrm{p}$ and $1.3\times10^{-3}\,\gamma_\mathrm{p}$ for the non-cooperative and cooperative receivers, respectively.

A stochastic simulation of the non-cooperative receiver with the random symbol stream at $\Omega=200$, using the same streams for every leak rate, measures the leak penalty directly. The ratio $P_\mathrm{e}(\gamma_\mathrm{c})/P_\mathrm{e}(0)$ is $0.97$ at $\gamma_\mathrm{c}=5\times10^{-3}\gamma_\mathrm{p}$, $1.16$ at $10^{-2}\gamma_\mathrm{p}$, and $1.96$ at $2\times10^{-2}\gamma_\mathrm{p}$, with 95\,\% confidence intervals of a few percent, and the ten-percent penalty is reached at $\gamma_\mathrm{c}=8.0\times10^{-3}\,\gamma_\mathrm{p}$ on the static channel. The leak lowers the mean sensed occupancy to $0.479$ at $2\times10^{-2}\gamma_\mathrm{p}$, and reduces the scaled modulator variance $\Omega\Var(\delta y)$ from $23.6$ to $20.4$, which partly compensates for the detuning. Hence, the Gaussian approximation is conservative, and the measured tolerance, which applies to the static channel and not to a time-varying channel, is the one carried into the physical design in Section~\ref{sec:bio}. For the cooperative receiver, whose error probability is too small to resolve a ten-percent change with 2000 trajectories, the same sweep bounds the penalty at $2\times10^{-2}\gamma_\mathrm{p}$ only by a factor of $1.77$.

Overall, the practical non-ideality of the design is leaky integration, not stability, and tolerating it requires the controller species to turn over about two orders of magnitude more slowly than the modulator.

\section{Numerical Results}\label{sec:results}
In this section, we evaluate the performance of the proposed receiver. Unless stated otherwise, the results are obtained by stochastic simulation of the reaction network~(\ref{eq:aif}) with the full nonlinear receptor map~(\ref{eq:hill}), driven by a stream of independent equiprobable bits with $T_\mathrm{sym}=0.5$, generated independently for every trajectory, such that the controller senses the instantaneous occupancy and performs the averaging itself. No linearization is used, except for the stability sweep with logistic sensing in Section~\ref{sec:stab} and the linearized-sensing check in Section~\ref{sec:res-hc}. High-copy ensembles ($\Omega=200$) are simulated by Poisson tau-leaping~\cite{Gillespie2001tau} with a step of $0.01$ time units ($0.005$ in the copy-number sweep), where the fraction of negative counts clipped to zero is below $10^{-4}$, and low-copy ensembles by the exact direct method~\cite{Gillespie1977exact}. Where both methods were applied, their error probabilities agree within the bootstrap intervals, and under the CIR-based model of Section~\ref{sec:res-tracking}, halving or quartering the step changes the error probabilities by at most $9\,\%$.

Every run that estimates an error probability discards a burn-in of at least $150$ time units, i.e., at least $15$ loop time constants, and records the modulator concentration at the symbol boundaries, where it is independent of the bit of the symbol that starts there. Hence, averaging the two conditional tails in~(\ref{eq:bep}) over the recorded ensemble is exact for a detector that samples at the start of a symbol, which is an idealization of the instantaneous Hill model. Conditioning on the transmitted bit at later detection phases, from $0.2$ to $1$ of the symbol, changes the error probability by at most $1.9$ to $9.5\,\%$, depending on the receiver. Confidence intervals are 95\,\% bootstrap intervals over trajectories. Runs in which the controller senses the symbol-averaged occupancy~(\ref{eq:g}) directly, i.e., without self-noise, are reported as \emph{averaged-input references}. The receivers of Table~\ref{tab:design} are simulated at their design points with an ideal integrator, $\gamma_\mathrm{c}=0$.

\subsection{Symbol Averaging, Channel Tracking, and the CIR-Based Model}\label{sec:res-tracking}
\begin{figure}[t]
\centering
\includegraphics[width=\columnwidth]{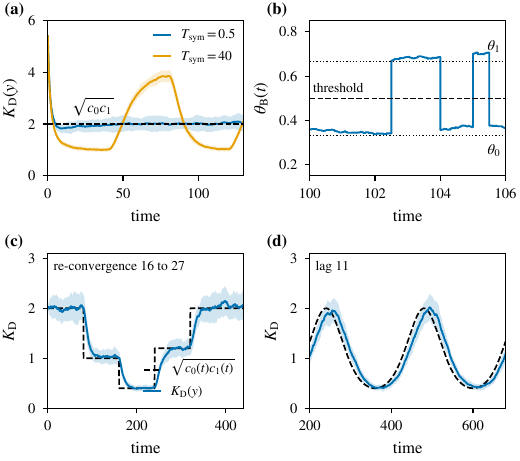}
\caption{Symbol averaging and channel tracking by the non-cooperative receiver ($\Omega=200$, 50 trajectories with independent symbol streams, shaded regions indicate one standard deviation). (a)~$K_\mathrm{D}$ for a random symbol stream with $T_\mathrm{sym}=0.5$ and an alternating symbol stream with $T_\mathrm{sym}=40$. (b)~Receptor occupancy of one trajectory under the random symbol stream. (c)~Response to attenuation steps applied every $T_\mathrm{c}=80$. (d)~Response to a sinusoidal attenuation of period $240$. Times are in model units.}
\label{fig:tracking}
\end{figure}
We first verify the two-timescale operation of the receiver. Fig.~\ref{fig:tracking}(a) shows the ensemble mean and standard deviation of $K_\mathrm{D}$ for the non-cooperative receiver, starting from an untuned state with $y=0$. Under a random bit stream with $T_\mathrm{sym}=0.5$, which is much shorter than the loop time constant, the affinity converges within approximately three loop time constants to $K_\mathrm{D}=2.03$, compared to $\sqrt{c_0c_1}=2.0$, and fluctuates about this value with a relative standard deviation of $14\,\%$, mostly due to the self-noise. As shown in Fig.~\ref{fig:tracking}(b), the occupancy alternates with the bits between $0.333$ and $0.664$, close to the values $1/3$ and $2/3$ predicted by~(\ref{eq:centering}). Hence, the controller holds the average occupancy, while the detector observes the full occupancy separation. For an alternating bit sequence with $T_\mathrm{sym}=40$, which exceeds the loop time constant, the loop instead retunes the affinity within each symbol and drives the occupancy toward the setpoint for both bits, to $0.50$ and $0.51$ by the end of the symbol, and thus removes the occupancy difference between the bits. Therefore, the first inequality in~(\ref{eq:window}) is a necessary condition for detection.

Fig.~\ref{fig:tracking}(c) shows the response of the receiver to attenuation steps by factors of two to three every $T_\mathrm{c}=80$. After each step, the affinity re-converges to the new optimum $\sqrt{c_0(t)c_1(t)}$ within $1.5$ to $2.7$ loop time constants. Under a sinusoidal attenuation with variation depth $S=5$ and period $240$, shown in Fig.~\ref{fig:tracking}(d), the ensemble-mean affinity follows the time-varying optimum with a lag of $1.1$ loop time constants, as expected from the tracking transfer function $T(s)$ of Section~\ref{sec:gain}. Its mean deviation from the optimum, which we refer to as the \emph{tracking error}, is $17\,\%$.

The diffusive tail of $h(t,d)$ adds ISI to every symbol. In the simulations, the tail is kept for $80$ symbols, and at the tested \emph{normalized symbol duration} $T_\mathrm{sym}/t_\mathrm{peak}=3$, it retains $88\,\%$ of the cumulative ISI of the free-diffusion pulse, i.e., the channel is a truncated diffusion channel with finite memory. All quantities of the experiments with ISI, including the setpoint, the best non-adaptive receiver, and a \emph{genie-aided receiver} that retunes its affinity exactly and without noise, are evaluated on the same truncated records. Under this ISI and the piecewise-constant signal model, the equilibrium affinity agrees to within one percent with the optimum $\sqrt{(c_0+\bar c_\mathrm{int})(c_1+\bar c_\mathrm{int})}$ predicted by the mean interference $\bar c_\mathrm{int}$~\cite{Kuscu2023adaptive}, provided that $K_\mathrm{base}$ is raised from $5.44$ to $20.00$ to bring this optimum within the range of the actuator.

The controller senses the receptor occupancy continuously, whereas the detector samples the occupancy at the peak of the diffusive pulse. As the piecewise-constant signal model does not capture this difference, we use a \emph{CIR-based model}, in which the received concentration is the superposition of the CIRs of the current and previous symbols, i.e., $c(t)=g(t)\sum_j N_{\mathrm{L}|b_j}h(t-jT_\mathrm{sym},d)$, where $b_j$ denotes the $j$th transmitted bit, and the CIR is normalized such that the isolated peak concentrations equal $c_0$ and $c_1$. As the concentration is below its peak during most of the symbol, the time average of the bit-averaged occupancy, evaluated at the affinity that is optimal for the peak samples, is below one half, and a controller with a setpoint of one half settles below that affinity. Under the strong ISI described above, with $g=1$ and $K_\mathrm{base}=20$, the affinity $K_\mathrm{peak}=5.80$, which minimizes the error probability of a fixed receptor sampled at the pulse peaks over the same interference histories, yields an error probability of $8.55\times10^{-2}$. The loop with a setpoint of one half, however, settles at $K_\mathrm{D}=4.79\pm0.05$, i.e., $17\,\%$ below $K_\mathrm{peak}$, with an error probability of $1.82\times10^{-1}$ at the peaks.

This mismatch can be repaired by exploiting the homogeneity of~(\ref{eq:hill}) in $(c,K_\mathrm{D})$. The waveform-averaged occupancy at $K_\mathrm{peak}$, which is $s^*=0.451$ for this waveform, depends on the normalized waveform and on the symbol statistics, but it is invariant under a common attenuation of the waveform. With the reference rate set to $\mu=s^*\theta$, the affinity settles at $5.84\pm0.06$, with an error probability of $1.10\times10^{-1}$, i.e., a factor of $1.28$ above that of the fixed receptor at $K_\mathrm{peak}$, and under an attenuation by a factor of $0.3$, it settles at $1.75$, compared to the optimum $1.74$. For an attenuation that varies with depth $S=5$ and period $T_\mathrm{c}=240$, the adaptive receiver with setpoint $s^*$ reaches an error probability of $1.72\times10^{-1}$, compared to $3.48\times10^{-1}$ for the best non-adaptive receiver and $8.55\times10^{-2}$ for the genie-aided receiver, i.e., a reduction of the BEP by a factor of $2.0$, which grows as the ISI weakens, as shown in Section~\ref{sec:res-headline}. Therefore, robust perfect adaptation to a common attenuation is retained under the CIR-based model, with a setpoint calibrated once to the waveform and the symbol statistics. However, a change of the transmitter-receiver distance alters the pulse shape as well as its amplitude, and the demonstrated autonomy applies only to the amplitude.

\subsection{Comparison with Non-Adaptive Reception}\label{sec:res-headline}
\begin{figure}[t]
\centering
\includegraphics[width=\columnwidth]{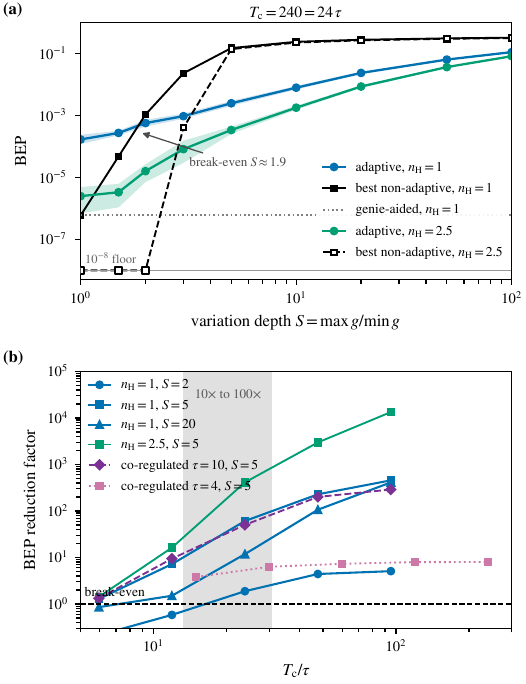}
\caption{Adaptive and non-adaptive receivers on the same time-varying channel with random symbol streams ($N_\mathrm{R}=200$, 48 trajectories, 192 for the cooperative receiver in (a), 95\,\% bootstrap intervals). (a)~BEP as a function of the variation depth $S$ of the channel for $T_\mathrm{c}=240$. (b)~BEP reduction factor as a function of $T_\mathrm{c}/\tau$, where the shaded band spans the coherence times over which the factor for $n_\mathrm{H}=1$ and $S=5$ rises from 10 to 100. Values below $10^{-8}$ are drawn at this floor with open markers.}
\label{fig:headline}
\end{figure}
Fig.~\ref{fig:headline} shows the benefit of adaptation on a time-varying channel, in which both received levels are attenuated by a factor that varies from one to $1/S$ and back with period $T_\mathrm{c}$, where $S$ denotes the \emph{variation depth} of the channel. We compare the adaptive receiver, which has a fixed threshold and no CSI, with the best non-adaptive receiver, whose fixed $K_\mathrm{D}$ and threshold are both optimized with full knowledge of the channel statistics, and with the genie-aided receiver, which marks the floor set by the receptor array. The compared receivers have the same cooperativity, and the non-adaptive benchmark of the co-regulated receivers is fixed at $n_\mathrm{H}=2.5$.

Fig.~\ref{fig:headline}(a) shows the results for $T_\mathrm{c}=240$, i.e., $24$ loop time constants. On a static channel ($S=1$), the non-adaptive receiver performs better, as predicted in Section~\ref{sec:bep}, by a factor of $270$ to $490$ for $n_\mathrm{H}=1$ over the coherence times considered. The adaptive non-cooperative receiver becomes superior at $S\approx1.9$, and the improvement is largest at $S=5$, with error probabilities of $2.5\times10^{-3}$ and $1.5\times10^{-1}$ for the adaptive and non-adaptive receivers, i.e., a factor of $61$. At larger variation depths, the error probability of the non-adaptive receiver rises above one quarter, while that of the adaptive receiver degrades gradually. The improvement nevertheless decreases to factors of $12$ at $S=20$ and $2.9$ at $S=100$, as the loop can no longer follow the entire excursion of the received levels within one period. A non-adaptive cooperative receptor with an optimized threshold, on the other hand, already covers a twofold variation, with an error probability of $5.2\times10^{-11}$ at $S=2$, compared to $1.6\times10^{-5}$ for the adaptive receiver. Hence, the adaptation of cooperative receptors becomes beneficial only from $S\approx2.9$, and its advantage exceeds that of the non-cooperative receiver from $S=5$ on, with factors of $410$ at $S=5$ and $130$ at $S=10$.

The rate of channel variation is as important as its magnitude. Fig.~\ref{fig:headline}(b) shows the ratio of the error probabilities of the best non-adaptive and the adaptive receivers, i.e., the \emph{BEP reduction factor}, as a function of the coherence time in units of the loop time constant. When the coherence time is only $6$ loop time constants, the non-cooperative receiver reduces the BEP by at most a factor of $1.3$. At $S=5$, its reduction factor grows from $7.2$ at $12$ time constants to $460$ at $95$, and that of the cooperative receiver from $16$ to $13000$. At $S=2$, the reduction factor of the non-cooperative receiver remains below $5.1$, whereas at $S=20$, it reaches $410$ at the longest coherence time. Therefore, the timescale condition~(\ref{eq:window}) is necessary but not sufficient. Under the piecewise-constant signal model, and for the tested range of variation depths, symbol rate, and copy number, the BEP reduction factor at $S=5$ reaches ten at
\begin{equation}
T_\mathrm{c}\approx13\,\tau\ (n_\mathrm{H}=1),\qquad T_\mathrm{c}\approx10\,\tau\ (n_\mathrm{H}=2.5),
\label{eq:rule}
\end{equation}
and one hundred at $31\,\tau$ and $18\,\tau$, respectively, as indicated by the shaded band in Fig.~\ref{fig:headline}(b). This empirical operating region holds for the design loop speed.

The fast co-regulated receiver, whose averaging window spans only $8$ symbols, reduces the BEP by at most a factor of $10$, since its self-noise dominates irrespective of the channel. The bandwidth-matched co-regulated receiver removes this noise penalty, with reduction factors of $50$ at $T_\mathrm{c}=240$ and $290$ at $T_\mathrm{c}=960$ for $S=5$, but although it is better than the non-cooperative receiver on the static channel, its error probabilities are $1.12\times$ and $1.46\times$ those of the non-cooperative receiver, and $8.2\times$ and $46\times$ those of the cooperative receiver. The reason is the attenuation dependence of its cooperativity identified in Section~\ref{sec:coop}. At the deepest point of the variation, $n_\mathrm{H}(y^*)$ rises to $4.11$, and the receptor gain falls from $0.319$ to $0.213$, whereas it is invariant under attenuation for a fixed $n_\mathrm{H}$. Hence, the loop is slowest where the channel is most attenuated, and its tracking error of $19\,\%$ at $T_\mathrm{c}=240$ exceeds the $17\,\%$ of the non-cooperative receiver. Therefore, co-regulation removes the noise penalty of the fast loop on the static channel, but forfeits the benefit of cooperativity on the time-varying channel at the tested variation depth.

\begin{figure}[t]
\centering
\includegraphics[width=\columnwidth]{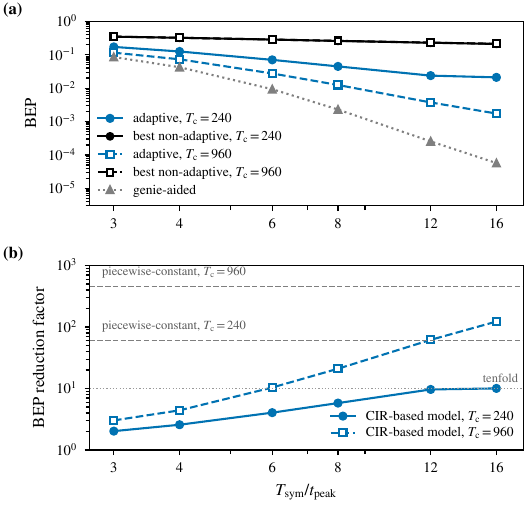}
\caption{Performance under the CIR-based model as a function of the normalized symbol duration $T_\mathrm{sym}/t_\mathrm{peak}$ for $S=5$ at $T_\mathrm{c}=240$ (filled markers) and $960$ (open markers), for the non-cooperative receiver with $K_\mathrm{base}=20$ at $\Omega=200$ (40 trajectories). (a)~BEP of the adaptive receiver, of the best non-adaptive receiver, and of the genie-aided receiver, which does not depend on the coherence time. (b)~BEP reduction factor, with the values of the piecewise-constant signal model at $S=5$ from Fig.~\ref{fig:headline} as dashed lines.}
\label{fig:isi}
\end{figure}
Under the CIR-based model, the benefit of adaptation also depends on the ISI. Fig.~\ref{fig:isi} repeats the comparison of Section~\ref{sec:res-tracking} at $T_\mathrm{c}=240$ and $960$ for normalized symbol durations from $3$ to $16$, i.e., from strong to weak ISI. The peak time is varied at a fixed symbol duration, and only the reference rate of the receiver is recalibrated, through $\mu=s^*\theta$, for each waveform. Since the pulse shape changes together with the ISI, the sweep compares complete operating conditions, instead of isolating the effect of the ISI. Over this range, the loop gain, evaluated with the waveform average $\langle\theta_\mathrm{B}(1-\theta_\mathrm{B})\rangle_t$ in place of $\theta_0\theta_1$ in~(\ref{eq:zeta}), stays within $81$ to $112\,\%$ of its design value, so that the actuator gain of the design can be retained. The relative affinity variance $\alpha^2\Var(\delta y)$ on the static channel stays within $0.0125$ to $0.0133$, compared to $0.0174$ under the piecewise-constant signal model.

As the ISI weakens, the error probability of the receptor with its affinity fixed at the static-channel optimum falls from $8.5\times10^{-2}$ to $5.8\times10^{-5}$, whereas that of the best non-adaptive receiver on the time-varying channel does not fall below $2.1\times10^{-1}$. Over the same range, the static-channel penalty of the closed loop grows from $1.3$ to $18$, as expected for an affinity fluctuation of constant size, whose relative cost grows as the error probability falls. Consequently, the BEP reduction factor grows with the normalized symbol duration. At $T_\mathrm{c}=240$, it grows from $2.0$ to $10$ and saturates above $12t_\mathrm{peak}$, where the adaptive receiver is limited by the lag of the loop instead of the ISI. At $T_\mathrm{c}=960$, where the lag is smaller, it grows from $3.0$ to $62$ at $12t_\mathrm{peak}$ and $123$ at $16t_\mathrm{peak}$, compared to $460$ under the piecewise-constant signal model. The tracking errors, of $16$ to $23\,\%$ at $T_\mathrm{c}=240$ and $4$ to $5\,\%$ at $T_\mathrm{c}=960$, are comparable with the $17$ and $4\,\%$ of the piecewise-constant signal model. Hence, the remaining gap between the reduction factors of the two signal models is mainly due to the ISI, which raises both received levels, thus compressing their ratio, and spreads the peak samples. As a result, even the genie-aided receiver reaches only $5.7\times10^{-5}$ at $16t_\mathrm{peak}$, compared to $6.3\times10^{-7}$ under the piecewise-constant signal model.

\begin{figure}[t]
\centering
\includegraphics[width=\columnwidth]{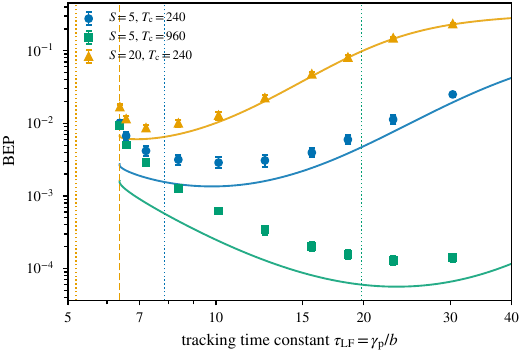}
\caption{Loop-speed design rule for the non-cooperative receiver. BEP as a function of the tracking time constant $\tau_\mathrm{LF}=\gamma_\mathrm{p}/b$ on the periodic channel, from simulation (markers, 48 trajectories, 95\,\% intervals), from the Gaussian approximation~(\ref{eq:surrogate}) (lines), and the leading-order candidate~(\ref{eq:tauopt}) (dotted verticals). For $S=20$ the unconstrained candidate lies below the smallest tracking time of the design family, marked by the dashed vertical.}
\label{fig:designrule}
\end{figure}
\begin{figure*}[t]
\centering
\includegraphics[width=\textwidth]{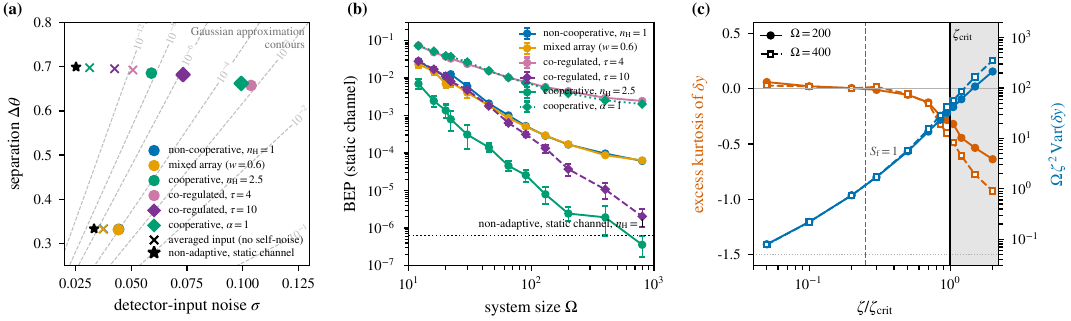}
\caption{(a)~Operating points of the receivers of Table~\ref{tab:design} (filled markers), of the averaged-input references (crosses), and of the non-adaptive receptor at the static-channel optimum (stars) at $\Omega=200$ in the plane of detector-input noise and occupancy separation. The contours are the Gaussian approximation $Q(\Delta\theta/2\sigma)$, which serves only as descriptive coordinates. The co-regulated markers use the bit-1 standard deviation. (b)~BEP as a function of the system size $\Omega$ (exact stochastic simulation for $\Omega\le30$, tau-leaping with 1000 trajectories above, 95\,\% intervals). (c)~Excess kurtosis of the modulator fluctuation and the modulator variance propagated to the occupancy, $\Omega\zeta^2\Var(\delta y)$, as a function of $\zeta/\zeta_\mathrm{crit}$ for the loop with logistic sensing of slope $\zeta$. The dashed line marks the normalized two-stage reference $S_\mathrm{f}=1$ of Section~\ref{sec:stab}.}
\label{fig:results}
\end{figure*}
Fig.~\ref{fig:designrule} tests the loop-speed rule of Section~\ref{sec:bep} by varying the actuator gain $\alpha$ of the non-cooperative receiver, and with it the tracking time $\tau_\mathrm{LF}=\gamma_\mathrm{p}/b$, which decreases with the actuator gain up to $\alpha=\sqrt2$, and thus identifies each design uniquely over the tested range, while $\tau$ saturates at $2/\gamma_\mathrm{p}$ once the reduced loop becomes underdamped. For $S=5$, the measured error probability is lowest at $\tau_\mathrm{LF}=10.1$ for $T_\mathrm{c}=240$ and $22.9$ for $T_\mathrm{c}=960$, i.e., the optimum moves by a factor of $2.27$, compared to $4^{2/3}=2.52$ from~(\ref{eq:tauopt}). The Gaussian approximation~(\ref{eq:surrogate}) places these minima at $9.8$ and $23.1$, and the rule~(\ref{eq:tauopt}) at $7.9$ and $19.8$, respectively. For $S=20$ at $T_\mathrm{c}=240$, the unconstrained candidate $5.2$ lies below the smallest tracking time $6.36$ of the design family, which is then the implementable candidate, compared to the measured optimum of $7.2$. Therefore, the projected rule underestimates the optimal tracking time by $12$ to $22\,\%$, which is the accuracy expected of a leading-order balance that neglects the phase average and the non-Gaussian tails, while it reproduces the dependence on the coherence time and the variation depth, and thus replaces the crossover heuristic used to locate the design points. Moreover, the minimum is shallow, as the error probability at the design point $\tau_\mathrm{LF}=12.6$ lies within a factor of $1.1$ of the measured minimum at $T_\mathrm{c}=240$ and $2.6$ at $T_\mathrm{c}=960$.

The actuation law also matters on a time-varying channel. In a separate ensemble at $T_\mathrm{c}=240$, we calibrated the bounded actuation law~(\ref{eq:mwc}) with $\chi=100$ to the loop gain of the non-cooperative design, with a range $[0.061,6.05]$ that contains the target affinities for variation depths up to $S=33$. On the static channel, the bounded and exponential laws are indistinguishable, with error probabilities of $1.5\times10^{-4}$ and $1.8\times10^{-4}$, respectively. At $S=5$, however, the bounded receiver reduces the BEP by a factor of only $6.7$, compared to $58$ for the exponential law in this ensemble, with a tracking error of $36\,\%$ compared to $17\,\%$, and at $S=20$, it performs worse than the non-adaptive receiver. The reason is not the limited range but the decreasing actuator gain, as reaching $K_\mathrm{D}=0.4$ at $S=5$ requires $y=12.0$, where $\beta(y)$ has fallen by a factor of $6.5$ below its calibration value, and the loop slows down accordingly. A single effector site thus does not maintain the actuator gain over the range of affinities required by the variation, unlike the exponential actuation law~(\ref{eq:exp}).

The exponential law, in turn, relies on the reduction of the modification level to the modulator concentration, which we checked by simulating the modification state $m$ explicitly, with $k_\mathrm{m}=k_\mathrm{d}/M$ and the actuator gain $\alpha$ of~(\ref{eq:exp}) preserved. With $M=100$ sites, the modification range $aM=38.4$ exceeds the shift $1+\ln S=2.61$ required by a fivefold variation by an order of magnitude, and for a removal rate of $k_\mathrm{d}=5\gamma_\mathrm{p}$, the BEP reduction factor of $48$ at $S=5$ and $T_\mathrm{c}=240$ agrees with the $61$ of the reduced model within the bootstrap intervals of the two ensembles. A slower removal, with $k_\mathrm{d}=0.2\gamma_\mathrm{p}$, adds a lag, which raises the scaled modulator variance from $25.5$ to $43.7$ and lowers the reduction factor to $36$. With $M=10$ sites, on the other hand, the modification range $aM=3.8$ is close to the required shift, and the local actuator gain $\alpha(1-m/M)^2$ falls as the modified fraction reaches $63.7\,\%$. As a result, the reduction factor at $S=5$ is only $2.7$, with a tracking error of $48\,\%$. Therefore, the model reduction holds for a modification range that is large compared to the required shift and a modification turnover that is at least as fast as that of the modulator. A small range, on the other hand, reduces the actuator gain where the variation requires it, as the bounded actuation law does.

\subsection{Performance on the Static Channel}\label{sec:res-static}
In this subsection, we evaluate the receivers on the \emph{static channel}, i.e., at $g=1$, which isolates the noise cost of adaptation from its tracking benefit. As the receiver has no CSI, it pays this cost on every channel, and in the tested cases, the error probability of the non-cooperative and co-regulated receivers on the time-varying channel never falls below its static-channel value, whereas the cooperative values at variation depths below two lie within the bootstrap intervals of their static-channel values.

\subsubsection{Comparison of Receptor Designs}\label{sec:res-hc}
Table~\ref{tab:design} and Fig.~\ref{fig:results}(a) compare the receivers at their design points at $\Omega=200$, where the non-adaptive receptor with its affinity fixed at the static-channel optimum marks the error-probability floor of the array. The exact binomial tail exceeds a Gaussian approximation with the measured means and standard deviations of the two occupancies by factors of $1.8$ for the non-cooperative and $930$ for the cooperative receiver, because the occupancy fluctuations are not Gaussian, and the wide occupancy separation of the cooperative receiver places the error deep in their tails. In every ensemble, the average occupancy is held at $0.5002$, and the mean affinity of the non-cooperative receiver is $2.016$, compared to $\sqrt{c_0c_1}=2.0$. The measured modulator variance lies $1$ to $15\,\%$ below the full linear-noise prediction, including the self-noise term. These deviations are not caused by the finite correlation time of the symbol sequence, but by the nonlinearity of the receptor map, since, with linearized sensing driven by the exact symbol switching, the fast co-regulated loop, which deviates the most, differs by only $-1\,\%$ from the linear prediction. Relative to the averaged-input references, the self-noise raises the modulator variance by factors of $2.8$ to $7.8$, depending on the receiver, and the error probability of the non-cooperative receiver by a factor of $16$.

The exact error probabilities in Table~\ref{tab:design} range from $2.5\times10^{-6}$ for the cooperative to $4.3\times10^{-3}$ for the fast co-regulated receiver, compared to $6.3\times10^{-7}$ for the non-adaptive receptor with $n_\mathrm{H}=1$ and its affinity fixed at the static-channel optimum, a difference that reflects the noise cost of adaptation discussed in Section~\ref{sec:bep}. Two further comparisons isolate the individual effects. First, a cooperative receptor with a fixed $n_\mathrm{H}=2.5$, operated at the actuator gain and rates of the fast co-regulated receiver, has an error probability within a factor of $1.1$ of the co-regulated receiver, and thus, the fluctuation of the cooperativity is not the cause of the co-regulation penalty. Second, the matched co-regulated receiver, with the same equilibrium cooperativity but the loop gain of the design, is $140\times$ better than the fast one. Its error probability is $0.19\times$ that of the non-cooperative receiver and $13\times$ that of the cooperative receiver, with a total of $1003$ controller and modulator molecules, close to the $1080$ of the non-cooperative receiver. The cooperative receiver uses $1580$ molecules, but at $\Omega=130$, where its molecule count of $1027$ lies slightly below $1080$, its error probability of $7.9\times10^{-6}$ is still $21\times$ below that of the non-cooperative receiver at $\Omega=200$. Hence, its advantage is not an artifact of the resource comparison.

\subsubsection{Dependence on the Copy Number}\label{sec:res-lc}
As the copy numbers of the controller species in a cell may be small, Fig.~\ref{fig:results}(b) shows the error probability as a function of the system size from $\Omega=12$ to $800$, with the cooperative receiver at $\alpha=1$ as a dotted line and the non-adaptive receptor with $n_\mathrm{H}=1$ as a horizontal line. In this separate ensemble, the cooperative receiver is the best at every copy number, with an error probability lower than that of the non-cooperative receiver by a factor of $7.9$ at $\Omega=20$ and $69$ at $\Omega=200$. The matched co-regulated receiver coincides with the non-cooperative receiver at the lowest copy numbers, with an error-probability ratio of $1.03$ at $\Omega=20$, and separates from it as the reaction noise recedes, with ratios of $0.22$ at $\Omega=200$ and $0.03$ at $\Omega=800$. The mixed-array receiver follows the non-cooperative receiver closely, and the fast co-regulated receiver and its fixed-cooperativity counterpart at $\alpha=1$ coincide at every $\Omega$, which confirms the conclusion of Section~\ref{sec:res-hc}.

The dependence on the copy number is determined by the two noise sources in~(\ref{eq:vardy}). Below $\Omega\approx100$, the reaction noise dominates, and the error probabilities fall steeply, e.g., by a factor of $65$ between $\Omega=20$ and $200$ for the non-cooperative receiver. Above this point, the self-noise term dominates, and the error probability of the non-cooperative receiver decreases by a factor of only $2.7$ between $\Omega=200$ and $800$. Hence, the sweep demonstrates diminishing returns over the tested range, although it does not establish the limit of the BEP as $\Omega\to\infty$. Longer averaging, if the channel permits, reduces the self-noise contribution, whereas cooperativity improves the other contributions and the resulting error probability, but leaves the normalized self-noise floor $\sigma_\mathrm{bits}^2/\Delta\theta^2$ unchanged at matched loop speed.

\section{Discussion on Biological Feasibility and Limitations}\label{sec:bio}
In this section, we discuss the biological feasibility of the receiver, provide an illustrative mapping to physical units, and state the assumptions of the model. One model time unit corresponds to 25~s, and thus, $T_\mathrm{sym}=0.5$ corresponds to 12.5~s, i.e., three peak times of the diffusive pulse at $d=50\,\mu$m and $D=100\,\mu$m$^2$/s~\cite{Kuscu2023adaptive}, and the loop time constant $\tau=10.1$ to 252~s.

\textbf{Chemical realization:} The controller reads the occupancy through the transduction pathway of the receptor, e.g., the production of cyclic adenosine monophosphate by adenylyl cyclase downstream of active G-protein-coupled receptors (GPCRs)~\cite{Klauer2024gpcr}, or the phosphorylation of CheY by the kinase CheA in bacterial chemotaxis~\cite{Sourjik2004functional}, and either second messenger equilibrates within a symbol. In an allosteric receptor, however, activity and occupancy are distinct, and the sensed quantity and the setpoint must then be calibrated together. The antithetic pair must annihilate in a one-to-one stoichiometry, as realized by a sigma factor and its anti-sigma factor in \emph{E.\ coli}~\cite{Aoki2019universal}, sense and antisense RNA in mammalian cells~\cite{Frei2022mammalian}, and split inteins~\cite{Anastassov2023intein}. The modulator $Y$ represents the activity of a modifying enzyme or of an effector, whose binding to the receptor must be fast compared to the loop. Equation~(\ref{eq:mwc}) describes the standard mechanism of allosteric modulation~\cite{Christopoulos2002allosterism,Canals2012mwc}, and intracellular allosteric sites on GPCRs have been characterized structurally~\cite{Liu2017intracellular}, which suggests that the effector can be a peptide or protein produced by the cell. Bacterial chemotaxis provides the closest natural analogy, as the methylation of the Tar receptor shifts the free energy of its active state~\cite{Keymer2006chemosensing,Mello2007effects} under integral feedback by the enzymes CheR and CheB~\cite{Yi2000robust,Alon1999robustness}. The analogy, however, holds only at the level of the control law, as CheR and CheB act catalytically and not by sequestration, and the adapted quantity is the kinase activity instead of the ligand occupancy.

\textbf{Timescales:} For a given symbol duration, the lower end of the admissible loop time constant is set by the self-noise. Averaging $\tau/T_\mathrm{sym}=20$ symbols, as required for the noise level of the design, fixes $\tau\approx252$~s and, via the design ratio $\tau\approx5/\gamma_\mathrm{p}$ of Table~\ref{tab:params}, a modulator turnover of $\gamma_\mathrm{p}\approx0.02$~s$^{-1}$. Transcriptional actuation, at $10^{-4}$ to $10^{-3}$~s$^{-1}$, gives loop time constants of $5\times10^{3}$ to $5\times10^{4}$~s and thus requires a channel that remains coherent for many hours, whereas phosphorylation cycling at $10^{-1}$ to $1$~s$^{-1}$ averages fewer than four symbols and raises the self-noise above the design level. Receptor methylation, at approximately $10^{-2}$~s$^{-1}$, gives $\tau\approx500$~s, close to the required value. The upper end is given by~(\ref{eq:rule}), which requires $T_\mathrm{c}\gtrsim13\,\tau$, i.e., $T_\mathrm{c}\gtrsim3277$~s, at $S=5$. Therefore, the modulator should be actuated by a post-translational modification with a turnover near $10^{-2}$~s$^{-1}$, while a transcriptional circuit may still provide the antithetic pair if its species turn over slowly enough for the integrator not to leak. This leak requirement of Section~\ref{sec:leaky} is demanding, as the tolerance of $8.0\times10^{-3}$ measured for a ten-percent increase in the error probability of the non-cooperative receiver corresponds to a doubling time longer than $1.2$~h, whereas the dilution rates of $10^{-4}$ to $10^{-3}$~s$^{-1}$ in growing cells correspond to $\gamma_\mathrm{c}/\gamma_\mathrm{p}=0.005$ to $0.05$. The design therefore requires slowly dividing or non-dividing cells, such as mammalian cells, or controller species protected from dilution. Under the CIR-based model, the BEP reduction factor at $T_\mathrm{c}=6000$~s is only $2.0$ under the strong ISI at $T_\mathrm{sym}=3t_\mathrm{peak}$, and grows to $62$ at $T_\mathrm{c}=24000$~s and $12t_\mathrm{peak}$, i.e., for $d=25\,\mu$m (Fig.~\ref{fig:isi}). Hence, condition~(\ref{eq:rule}) does not transfer unchanged to the CIR-based model.

\textbf{Assumptions and limitations:} The cytosol of a cell of radius $R=1\,\mu$m mixes within $R^2/D_\mathrm{cyt}\sim0.1$~s for $D_\mathrm{cyt}\approx10\,\mu$m$^2$/s, three orders of magnitude below the loop time constant, which justifies the well-mixed treatment of the network. Binding is assumed to be at quasi-steady state within a symbol, which requires $k_-\gg1/T_\mathrm{sym}$ and holds for typical unbinding rates of $1$ to $10^{2}$~s$^{-1}$, and the receptor units are independent. The channel variation is modeled as a common attenuation of both received levels without pulse-shape changes, the CIR-based model is restricted to the free-diffusion CIR, and the analysis treats the interference through its mean, whereas the simulations use explicit pulse histories truncated to $80$ symbols. The counting noise of the channel is not included, and the sensing and actuation interface is idealized, as the controller senses the occupancy without binding fluctuations or second-messenger noise, and the modification state evolves deterministically. The tolerance of the noisy receiver to an error in the setpoint ratio $\mu/\theta$, which detunes the receptor affinity, is not quantified here. The model rates are nondimensional and mapped to physical ranges, not to a specific cell, and $\Omega=200$ corresponds to $1080$ controller and modulator molecules for the non-cooperative receiver. The reported error probabilities are therefore results of the model, and not measurements of receiver performance.

\section{Conclusion}\label{sec:conclusion}
In this paper, we investigate an MC receiver that autonomously tunes the affinity of its ligand receptors through an intracellular chemical reaction network. We show that a four-reaction AIF controller, which holds the symbol-averaged occupancy at one half, places the dissociation constant, without any CSI, at the geometric mean of the received levels, which, for a piecewise-constant received signal and equal priors, minimizes the error probability of the fixed-threshold detector for any fixed receptor cooperativity. The self-noise induced by the random symbol stream sets an error floor that does not decrease with the copy numbers of the network molecules, and balancing it against the channel-tracking error yields a loop-speed rule that the simulations support to within about a quarter of the optimal tracking time. Neglecting the self-noise, as in the averaged-input idealization, would make the non-cooperative receivers appear an order of magnitude better than they are and the cooperative ones fall below the reporting threshold of $10^{-8}$. For a variation depth of $S=5$ and a coherence time of $24$ loop time constants, adaptation lowers the error probability by factors of $61$ and $410$ for non-cooperative and cooperative receptors, respectively, relative to the best non-adaptive receiver, a benefit that grows with the coherence time, peaks at moderate variation depths, and, under the CIR-based model, rises from a factor of $2.0$ to $123$ as the ISI weakens and the coherence time increases. On a static channel, however, a non-adaptive receiver remains better. Receptor cooperativity widens the occupancy separation by more than its noise cost, whereas co-regulated cooperativity constrains the actuator gain and slows the loop on a time-varying channel. The actuator must retain its gain over the range of affinities required by the variation depth, unlike the tested bounded actuation law and saturating modification, and the illustrative mapping to physical units calls for post-translational actuation and for controller species diluted more slowly than in fast-growing cells. Natural extensions are a sensor-based proportional-integral controller~\cite{Filo2026sensor}, which reduces the modulator noise without additional species, a modulator with several modification sites, and a physically parameterized channel with counting noise and a mobile transmitter. These results indicate that the adaptivity of natural sensory systems can be transferred to engineered MC receivers with a small number of reactions, and provide the design conditions under which such receivers are beneficial.

\section*{Data and Code Availability}
{\sloppy
The simulation and figure-generation code, together with a table of the settings of every reported experiment, is available at \url{https://github.com/calico-lab/autonomous-receptor-tuning}.\par}

\bibliographystyle{IEEEtran}
\bibliography{references}

\end{document}